\documentclass[12pt,reqno]{amsart}

\usepackage[OT1]{fontenc}
\usepackage{amsthm,amsmath,amsfonts,amssymb,amsaddr}
\usepackage{hyperref}

\usepackage{booktabs, multirow}
\usepackage{lmodern}

\usepackage{mathrsfs}
\usepackage{ascmac}
\usepackage{bm}
\usepackage{natbib}
\usepackage{ascmac}
\usepackage{fancybox}
\usepackage{graphicx}
\usepackage{xcolor}
\usepackage{tabularx}
\usepackage{subcaption}
\usepackage{algorithm}
 \usepackage{algpseudocode}

\usepackage{amsaddr}
\usepackage{fancyhdr}
\usepackage{natbib}
\usepackage{fullpage}
\usepackage{newtxtext,newtxmath}
\usepackage{comment}

\newtheorem{theorem}{Theorem}

\newtheorem{lemma}{Lemma}
\newtheorem{proposition}{Proposition}

\theoremstyle{definition}
\newtheorem{definition}{Definition}

\newcommand{\R}{\mathbb{R}}

\newcommand{\mD}{\mathcal{D}}

\newcommand{\mK}{\mathcal{K}}
\newcommand{\mS}{\mathcal{S}}
\newcommand{\mT}{\mathcal{T}}

\newcommand{\mM}{\mathcal{M}}

\renewcommand{\hat}{\widehat}
\renewcommand{\tilde}{\widetilde}

\newcommand{\argmin}{\operatornamewithlimits{argmin}}
\newcommand{\argmax}{\operatornamewithlimits{argmax}}

\newcommand\NoDo{\renewcommand\algorithmicdo{}}
\newcommand{\given}{\,|\,}
\newcommand{\T}{\mathrm{\scriptscriptstyle T}}

\newenvironment{psmallmatrix}
  {\left(\begin{smallmatrix}}
  {\end{smallmatrix}\right)}

\title{Process-based Inference for Spatial Energetics \\ Using Bayesian Predictive Stacking}

\author{Tomoya Wakayama$^\dagger$ \and Sudipto Banerjee$^{\ddagger}$}
\address{$^\dagger$The University of Tokyo / $^\ddagger$ University of California, Los Angeles}

\date{\today}

\begin{document}

\begin{abstract}
 Rapid developments in streaming data technologies have enabled real-time monitoring of human activity that can deliver high-resolution data on health variables over trajectories or paths carved out by subjects as they conduct their daily physical activities. Wearable devices, such as wrist-worn sensors that monitor gross motor activity, have become prevalent and have kindled the emerging field of ``spatial energetics'' in environmental health sciences. We devise a Bayesian inferential framework for analyzing such data while accounting for information available on specific spatial coordinates comprising a trajectory or path using a Global Positioning System (GPS) device embedded within the wearable device. We offer full probabilistic inference with uncertainty quantification using spatial-temporal process models adapted for data generated from ``actigraph'' units as the subject traverses a path or trajectory in their daily routine. Anticipating the need for fast inference for mobile health data, we pursue exact inference using conjugate Bayesian models and employ predictive stacking to assimilate inference across these individual models. This circumvents issues with iterative estimation algorithms such as Markov chain Monte Carlo. We devise Bayesian predictive stacking in this context for models that treat time as discrete epochs and that treat time as continuous. We illustrate our methods with simulation experiments and analysis of data from the Physical Activity through Sustainable Transport Approaches (PASTA-LA) study conducted by the Fielding School of Public Health at the University of California, Los Angeles.
\end{abstract}

\maketitle

\section{Introduction}\label{sec:intro}
Spatial energetics is a rapidly emerging area in biomedical and health sciences that aims to examine how environmental characteristics, space, and time are linked to activity-related health behaviors \citep{JAMES2016792,Burgi2016}. Examples include, but are not limited to, using data from wearable devices as biomarkers and risk factors in studying adverse health outcomes for respiratory health \citep[][]{Kim2024}. Inferential objectives for spatial energetics comprise two exercises: (i) estimate measured health variables, typically related to metabolic activities, over paths or trajectories traversed by subjects while conducting their daily physical activities; and (ii) predict health variables for a subject at arbitrary trajectories. Spatial-temporal process models seem a natural choice as they use space-time coordinates from Global Positioning Systems (GPS) embedded in actigraph units.

Some salient features of spatial energetics require consideration. Unlike in some clinical studies associated with mobile health data, where only the temporal nature of streaming data is of inferential interest, here inferential interest centers around estimation and prediction of metabolic measurements over arbitrary spatial trajectories or paths. This differs from customary geostatistics and spatial-temporal data analysis \citep[see, e.g.,][and references therein]{banerjee2014hierarchical, cresswikle2015statistics} where statistical inference proceeds from spatial-temporal processes $\{w(s,t) : s\in \mS,\ t\in\mT\}$, where $\mS \subset \mathbb{R}^d$ with $d=2$ or $3$ and $\mT \subset \mathbb{R}^{+}\cup\{0\}$. For mobile health applications, the spatial domain is typically an arbitrary string of spatial coordinates defining the path or trajectory traversed by the subject. This trajectory is completely arbitrary and need not enjoy mathematically attractive features as are available for Riemannian manifolds to carry out inference \citep{liManifoldsGP2023}. Importantly, treated as continuously evolving over time, the spatial coordinates are functions of time.

The emerging literature on actigraph data has largely focused on temporal modeling \citep{crouter2006estimating,chang2022empirical, luo2023streaned, banker2023accelerometer} and have excluded spatial attributes of trajectories. Spatial extensions seek stochastic processes using either dynamic models \citep[][]{stroud2001dynamic, pena2004forecasting, gamerman2008spatial} or continuous processes using covariance kernels. Gaussian processes are conspicuous, but have largely focused on Euclidean domains. Actigraph data, however, arise from arbitrary trajectories that do not offer covariance kernels with exploitable geometric properties. While inference on non-Euclidean domains is possible \citep[][]{hoef2006spatial, ver2010moving, santos2022bayesian}, they do not consider spatial coordinates as functions over time. This also preclude predicting metabolic variables at arbitrary space-time coordinates, which is central to our aims. 

We proffer our specific contributions here. While related, our current developments differ from \cite{hooten2017animal}, who modeled spatial-temporal animal movement, and \cite{jurekMobilityFlightPause}, who analyzed movement using flight-pause models. Such models seek inference on the evolving path itself, whereas we analyze data that are collected at high resolutions on subjects moving along trajectories. Unlike studies over fixed geographic structures, where the domain of interest is a fixed graph and one does not pursue inference beyond a fixed structure, we seek predictive inference on metabolic activity would subjects adopt posited new trajectories for their physical activity. We devise Bayesian inference, including predictions over trajectories, with fully model-based uncertainty quantification for mobile data. 
Rather than opting for excessively complex non-linear models and potentially cumbersome Markov chain Monte Carlo (MCMC) (with which we compare) or other iterative algorithms such as Integrated Nested Laplace Approximations or Variational Bayes \citep[see, e.g.,][]{Robert2004, rue2009approximate, gelman2013bayesian,pml2Book} for fitting them, we extend predictive stacking \citep{yao2018using,bhatt2017improved,zhang2023exact} for fast Bayesian inference on actigraph data from wearables. 

Here, we overcome two specific challenges: (i) we avoid degenerate stochastic processes that would normally result from a subject visiting the same location multiple times (typical in mobile health data from walking, running); and (ii) accommodating spatial-temporal dependence accounting for locations being functions of time. Notably, we stack over analytically tractable posterior distributions to achieve fast and robust inference, which also renders our learning framework suitable for artificially intelligent ``recommender'' systems based upon an individual's specific health attributes.

Section~\ref{sec:review} discusses Bayesian hierarchical models that treat space as continuous and time as discrete. In particular, we show how to adapt Bayesian dynamic linear models \citep{west1997bayesian, prado2021time, stroud2001dynamic} to actigraph data. Sections~\ref{sec:con-con} develops a spatial-temporal process model \citep{cressie1999classes, stein2005space, gneiting2002nonseparable} treating both space and time as continuous. Section~\ref{section:stacking} develops predictive stacking by averaging over sets of conjugate Bayesian models with accessible posterior distributions. Section~\ref{sec: theory} collects some results on distribution theory and offers theoretical insights. Sections~\ref{section:simulation}~and~\ref{section:application} present simulation experiments and an illustrative application, respectively. Section~\ref{section:discuss} concludes the article with a discussion and pointers to future research.

\section{Continuous space and discrete time models}\label{sec:review}
Spatial-temporal models are classified based on continuous or discrete processes for space and time. Bayesian dynamic linear models \citep[DLMs,][]{west1997bayesian,prado2021time} analyze temporal data using discrete time steps. Space is modeled as a continuous random field evolving over time. We first show how such models can be adapted to actigraph data after which we consider spatial-temporal covariance kernels for actigraph data.  

\subsection{Spatial-temporal Bayesian DLMs}\label{sec: bayesian_dlm}
Actigraph data, while by nature stream in a continuum, are often recorded over a discrete set of epochs. Each epoch consists of a time-interval that can range from a few seconds to hours, or even days, depending upon the application. Let $\mT = \{1,2\ldots,T\}$ be a finite set of labels for epochs and $y_t$ be an $n_t\times 1$ vector consisting of measurements recorded by the actigraph at time $t$. A fairly flexible process-based model posits that $y_t = X_t\beta_t + z_t + \eta_t$ for each $t\in \mT$, where $X_t$ is an $n_t\times p$ matrix of explanatory variables, $\beta_t$ is the corresponding $p\times 1$ vector of slopes that depend on $t$ and $z_t$ is $n_t\times 1$ consisting of random effects accounting for other extraneous effects at time $t$. We construct the Bayesian DLM as  
\begin{equation}\label{model:dlm}
    \begin{split}
        {y}_t &= F_t\theta_t + {\eta}_{t},\quad {\eta}_{t} \overset{\mathrm{ind}}{\sim} N({0},\sigma^2 V_t);\quad \theta_t = G_t\theta_{t-1} + \eta_{\theta, t},\quad \eta_{\theta,t} \overset{\mathrm{ind}}{\sim} N(0, \sigma^2 S_t)\;,
    \end{split}
\end{equation}
where $F_t = \begin{psmallmatrix} X_t & I_n \end{psmallmatrix}$ is $n_t\times (p+n_t)$, and $\theta_t = \begin{psmallmatrix} \beta_t^{\T} & z_t^{\T}\end{psmallmatrix}^{\T}$ is $(p+n_t)\times 1$. We specify the prior distributions $\sigma^2\sim IG(n_{\sigma}/2, n_{\sigma}s_{\sigma}/2)$ and ${\theta}_0 \mid \sigma^2 \sim N({m}_{0}, \sigma^2S_0)$ so that the joint distribution is from the Normal-IG family. The quantities $G_t$, $n_{\sigma}$, $s_{\sigma}$, ${m}_{0}$ and $S_{0}$ are constants, while $\{{\theta}_t,\sigma,V_t,S_t\}$ are unknown parameters.

Two adaptations are relevant for spatial energetics. First, $y_t$ consists of $n_t$ measurements recorded at epoch $t$ independently over a group of subjects. The collected measurements $\{y_t : t\in \mT\}$ constitute an actigraph time-sheet, where each epoch provides values for the elements of $X_t$. Our available data, therefore, is $\{(y_t, X_t) : t \in \mT\}$. While each epoch implicitly contains information on the spatial locations for the subjects' measurements, the inferential goals for these population-level studies do not entail spatial attributes and, instead, are concerned with inferring about relationships between metabolic measurements (representing levels of physical activity) and environmental variables (representing green spaces, climate and weather, local topography, nature of activity being performed by the subject, etc.). We reasonably assume that $V_t$ is diagonal since measurements on subjects are independent and shared features across subjects are accounted for in explanatory variables and random effects in $F_t$. The covariance matrix for the elements of $\theta_t$, $S_t$, is assumed to be diagonal if latent associations are adequately accounted for by $G_{t}$. Alternative models could specify $S_t$ from design considerations or be modeled using an appropriate prior distribution \citep[see, e.g.,][]{west1997bayesian, prado2021time}. 

The second adaption of \eqref{model:dlm} applies to actigraph data on a single subject. Now $y_t(s)$ represents the rendered value of a variable of interest observed on a given subject at epoch $t\in \mathcal{T}$ and $s$ is the spatial coordinates of the subject at that epoch. Our process model specifies $y_t(s) = {x}_t(s)^{\T} {\beta}_t + z_t(s) + \eta_{t}(s)$, where ${x}_t(s)$ is a $p\times 1$ vector consisting of $p$ explanatory variables, ${\beta}_t$ is the corresponding $p\times 1$ vector of time-varying slopes, $z_t(s)$ is a zero-centered stochastic process, and $\eta_{t}(s) \overset{\mathrm{i.i.d.}}{\sim} N(0,\sigma^2)$. Therefore, ${x}_t^{\T}(s){\beta}_t$ represents the time-varying trend while $z_t(s)$ models temporal evolution with spatial dynamics. Let $\chi = \{s_i \mid i=1,\ldots,n \} \subset \mathbb{R}^2$ be a finite set of distinct spatial locations, where $y_t(s)$ has been measured. Then, $y_t$ and $z_t$ are $n\times 1$ vectors with elements $y_t(s_i)$ and $z_t(s_i)$, respectively, $X_t$ is $n\times p$ with rows $x_t(s_i)^{\T}$, and $\displaystyle S_t = \begin{psmallmatrix} \delta_{\beta}^2I_p &  O \\ O &  \delta_{z}^2 K_{\phi}(\chi)  \end{psmallmatrix}$ with $K_{\phi}(\chi) = (K_{\phi}(s_i,s_j))_{s_i,s_j\in \chi}$ is $n\times n$ with elements $K_{\phi}(s_i, s_j)$ evaluated using a spatial correlation kernel with parameters $\phi$, and $\delta_z^2$ and $\delta^2_{\beta}$ act as relative variance scales with respect to $\sigma^2$. 

The model in \eqref{model:dlm} delivers inference through MCMC or forward filtering-backward sampling \citep{carter1994gibbs,fruhwirth1994data}. However, with high-dimensional parameters, these methods require substantial computational resources that render their practical application to be challenging or even infeasible. Consequently, we devise Bayesian predictive stacking that exploits analytically tractable posterior distributions.

\subsection{Dynamic trajectory model}\label{sec:con-dis}
A salient feature of actigraph data encoded with spatial positioning is that the locations themselves are functions of time. \cite{chang2022empirical} and \cite{alaimo2023bayesian} have, therefore, modeled mobile data as processes primarily evolving over time with the latter accounting for spatial variation using splines. Models that introduce spatial-temporal associations will need to construct processes over the collection of points $\{(\gamma(t),t) : t \in \mathcal{T}\}$, where $\mathcal{T} \subset \mathbb{R}^{+}\cup \{0\}$ and $\gamma(t) : \mT \rightarrow \mathbb{R}^2$. Consider a single subject who has worn an actigraph unit that has recorded measurements at each time point $t$. Typically, such data are received as averages over discrete epochs so we define our temporal domain $\mathcal{T} = \{1,2,\ldots,T\}$ as a finite set of epochs spanning the entire duration of data collection from the device. Designing the data collection from a wearable device is, by itself, a meticulous exercise that needs to account for various extraneous factors including, but not limited to, the technologies of accelerometers as well as the specific clinical study under consideration \citep[see, e.g.,][for one such case study]{alaimo2023bayesian}. Here, we will concern ourselves with Bayesian inference.

Let $y_{t}(\gamma(t))$ denote the measurement on a given subject at time $t$ and location $\gamma(t)$ for each $t=1,\ldots,T$, $\gamma(t)\in \mathbb{R}^2$, and consider the following regression model:
\begin{equation}\label{model:condis}
    y_t( \gamma(t) ) = {x}_t(\gamma(t))^{\T} {\beta}_t + w_t(\gamma(t)) + \eta_{1t}(\gamma(t)), \quad \eta_{1t}(\gamma(t)) \overset{\mathrm{i.i.d.}}{\sim} N(0,\sigma^2)\;,
\end{equation}
where ${x}_t(\gamma(t))$ is a $p$-dimensional explanatory variable, ${\beta}_t$ is a $p$-dimensional time-varying regression coefficient, $w_t(\gamma(t))$ and $\eta_{1t}$ are zero-centered spatial and white noise processes, respectively, at time $t$, and $\sigma^2$ is the variance of the white noise process. The domain of the process $w_t(\cdot)$ is not Euclidean, but an arbitrary trajectory defined by a string of coordinates mapped by $\gamma(t)$. Furthermore, a subject may revisit the same location multiple times, yielding multiple values at $\gamma(t)$, making the spatial covariance matrix singular and, therefore, precluding legitimate probabilistic inference. We obviate this as follows.

Let $\Gamma = \{\gamma(1),\ldots,\gamma(T) \}$ be a complete enumeration of spatial locations visited by the subject, of which $n \leq T$ are distinct spatial locations, and let $\tilde{\Gamma} = \{\tilde{\gamma}_1,\ldots,\tilde{\gamma}_n\} \subseteq \Gamma$ be the subset of distinct locations. Let $z_t(\gamma(t))$ be a latent process and map $w_t(\gamma(t)) = \sum_{j=1}^n b(\gamma(t), \tilde{\gamma}_j) z_t(\tilde{\gamma}_j)$, where $b(\gamma(t),\tilde{\gamma}_j) : \Gamma \times \tilde{\Gamma} \to \{0,1\}$ such that $b(\gamma(t),\tilde{\gamma}_j) = 1$ if $\gamma(t) = \tilde{\gamma}_j$ and $0$ otherwise. This yields $w = Bz$, where $w = (w_1(\gamma(1)),\ldots,w_T(\gamma(T)))^{\T}$ is $T\times 1$, ${z} = ({z}_1^{\T},\ldots,{z}_T^{\T} )^{\T}$ is $nT\times 1$ with each ${z}_t =({z}_t(\tilde{\gamma}_1),\ldots,{z}_t(\tilde{\gamma}_n) )^{\T}$ being $n\times 1$, and $B$ is $T\times nT$, whose $(t,n(t-1)+j)$th entry is $b(\gamma(t),\tilde{\gamma}_j)$. This maps the latent spatial effects, $z$, to those at the observed points, $w$. If ${\beta} = ( {\beta}_1^{\T},\ldots, {\beta}_T^{\T} )^{\T}$, then temporal autoregressive models for $\beta$ and $z$ are specified as ${\beta} = (A \otimes I_p) {\beta} + \eta_{2t}$ and $z = (A\otimes I_n)z + \eta_{3t}$, respectively, where $\eta_{2t} \overset{\mathrm{i.i.d.}}{\sim} N({0},\sigma^2\delta_{\beta}^2 (I_T \otimes W_p))$ and $\eta_{3t} \sim N(0,\sigma^2\delta_z^2( I_T\otimes K_{\phi}))$, $A = 
\begin{psmallmatrix} {0}^{\T} &  0 \\ 
I_{T-1} & 0 \\
\end{psmallmatrix} \in \mathbb{R}^{T\times T}$, $W_p\in \mathbb{R}^{p\times p}$ is a correlation matrix among the coefficients and $\otimes$ denotes the Kronecker product.  We construct the following augmented model,
\begin{equation} \label{model:lin_sys_dis}
    Y = X\theta + \eta, \quad {\eta}\sim N({0}, \sigma^2S)\;,
\end{equation}
where $Y = \begin{psmallmatrix} y \\ 0 \\ 0\end{psmallmatrix}$ is $(1+p+n)T\times 1$, $X = \begin{psmallmatrix} 
\oplus_{t=1}^T{x}_t(\gamma(t))^{\T} & B \\ I_{pT}  &  O  \\ O &  I_{nT} \end{psmallmatrix}$ is $(1+p+n)T\times (p+n)T$, $\theta = \begin{psmallmatrix} \beta \\ z \end{psmallmatrix}$ and $\oplus$ is the block-diagonal matrix operator so $\oplus_{t=1}^T({x}_t(\gamma(t))^{\T})$ is $T\times pT$ block diagonal with ${x}_t(\gamma(t))$'s along the diagonal. Furthermore, $S = I_T \oplus \{ \delta_{\beta}^2 (I_{pT} - A\otimes I_p )^{-1} (I_T\otimes W_p)(I_{pT} - A^{\T}\otimes I_p )^{-1}\} \oplus \{\delta_{z}^2 (I_{nT}-A\otimes I_n)^{-1} (I_T\otimes K_{\phi}) (I_{nT}-A^{\T}\otimes I_n)^{-1}\}$. We introduce the prior distribution $\sigma^2\sim IG(a_{\sigma}, b_{\sigma})$, where $a_{\sigma}$ and $b_{\sigma}$ are fixed rate and scale parameters for the inverse-Gamma distribution. We assume that $W_p$ is assumed to be known and taken as the identity matrix in the later experiments. The prior distribution for $\theta$ is absorbed into \eqref{model:lin_sys_dis} and fixing the values of $\{\delta_{\beta}, \delta_z, \phi\}$ yields the familiar Normal-IG conjugate posterior distribution for $\{\theta,\sigma^2\}$, which is utilized in predictive stacking.  

\section{Continuous space-time trajectory model}\label{sec:con-con}

We can treat actigraph data as a partial realization of a continuous spatial-temporal process. We write $y(\gamma(t),t)$ to be the measurement that can exist, conceptually, at any time $t \in \mathbb{R}^{+}$ and spatial location $\gamma(t) \in \mathbb{R}^2$ at $t$. We construct a regression model.
\begin{equation}\label{model:reg_con}
    y(\gamma(t),t) = {x}(\gamma(t),t)^{\T} {\beta}(t) + z(\gamma(t),t) + \eta_1(t), \quad \eta_1(t) \overset{\mathrm{i.i.d.}}{\sim} N(0,\sigma^2), 
\end{equation}
over a finite set of $n$ space-time coordinates $(\gamma(t),t)$, where ${x}(\gamma(t),t)$ is $p\times 1$ consisting of explanatory variables, ${\beta}(t)$ is the corresponding $p\times 1$ vector of slopes, $z(\gamma(t),t)$ is a zero-centered spatial-temporal process and $\eta_1(t)$ is the measurement error distributed as a zero-centered Gaussian distribution with variance $\sigma^2$. 

For the spatial-temporal process, we consider the following structure:
\begin{equation}\label{model:st_process}
    z(\gamma(t),t)  \overset{\mathrm{ind}}{\sim} GP(0, \sigma^2\delta_{z}^2K_{\phi} ),
\end{equation}
where $K_{\phi}$ is a valid spatial-temporal correlation kernel \citep{cressie1999classes, gneiting2002nonseparable, stein2005space}. Among rich classes of functions, we work with a non-separable function,
\begin{equation}\label{eq: stkernel}
    K_{\phi}((\gamma(t),t), (\gamma(t'),t') ) = \frac{1}{ \phi_1 |t-t'|^2+1 } \exp\left(  -\frac{\phi_2 \| \gamma(t)-\gamma(t') \| }{\sqrt{1+\phi_1 |t-t'|^2} }   \right),\quad \phi_1,\phi_2\in\mathbb{R}^{+}\;,
\end{equation}
which is a special case of a more general class in \cite{gneiting2002nonseparable}. For each regression coefficient, we consider the following process:
\begin{equation}\label{model:t_process}
        {\beta}_j(t) \overset{\mathrm{ind}}{\sim} GP(0, \sigma^2\delta_{\beta}^2C_{\xi} ), \quad \mathrm{for} \quad j=1,\ldots,p, 
\end{equation}
with temporal correlation kernel $C_{\xi}(t,t') = \exp(-\xi^2|t-t'|^2)$. We note that $K_{\phi}$ is a positive-definite kernel ensuring that the stochastic process~\eqref{model:st_process} is well-defined. In particular, it is crucial to note that even when $\gamma(t)=\gamma(t')$, i.e., the subject returns to the same location at a later time, the function $(1+\phi_1 |t-t'|^2)^{-1}$ is positive-definite. This is seen from
\begin{equation*}
    ( 1+\phi_1 |t-t'|^2)^{-1} =\int_{0}^{\infty} e^{- u ( 1+\phi_1|t-t'|^2) } du.
\end{equation*}
Since $|t-t'|^2$ and $1$ are conditionally negative-definite, the integrand is positive-definite from Schoenberg's theorem \citep[e.g.,][]{phillips2019extension}, which implies that $(1+ \phi_1 |t-t'|^2)^{-1}$ is also positive-definite. If we observe $y$ at space-time points $(\Gamma,\mT) = \{ (\gamma(t_i), t_i) \mid i=1,\ldots,n \} \subset \mathbb{R}^2\times\mathbb{R}^{+}$, then (\ref{model:reg_con}),(\ref{model:st_process}) and (\ref{model:t_process}) for $n$ space-time points yields:
\begin{equation} \label{model:lin_sys_con}
    \underbrace{\begin{pmatrix}{y} \\ {0} \\ {0} \end{pmatrix}}_{Y} = 
    \underbrace{\begin{pmatrix} 
    [{x}_1\mid \cdots \mid {x}_p ] &  I_n \\
    [I_n \mid \cdots \mid I_n ] &  O  \\
        O &  I_n 
    \end{pmatrix}}_{X}
    \underbrace{\begin{pmatrix}{\beta}_1 \\ \vdots \\ {\beta}_p \\ {z}
   \end{pmatrix}}_{{\theta}} + {\eta},\quad {\eta}\sim N({0}, \sigma^2 S)
\end{equation}
where each ${x}_j\in\R^{n\times n}$ is diagonal with entries $x_j(\gamma(t_i),(t_i))$ for $i=1,\ldots,n$ and $j=1,\ldots,p$; $X$ is $3n \times (p+1)n$, ${\theta}$ is $(p+1)n\times 1$ consisting of the $pn\times 1$ vector of regression coefficients $\beta = (\beta_1^{\T},\ldots,\beta_p^{\T})^{\T}$ and the $n\times 1$ vector $z = (z(\gamma(t_1),t_1),\ldots,z(\gamma(t_n),t_n))^{\T}$, $K_{\phi}(\Gamma,\mT) = (K_{\phi}((\gamma(t_i),t_i), (\gamma(t_j),t_j) ))$ and $C_{\xi}(\mT) = (C_{\xi} (t_i,t_j))$ are both $n\times n$, where $i,j=1,\ldots,n$ and $\displaystyle S = I_n \oplus\left(I_p\otimes \delta_{\beta}^2 C_{\xi}(\mT)\right) \oplus \delta_{z}^2 K_{\phi}(\Gamma,\mT)$. We further assign the prior distribution $\sigma^2\sim IG(a_{\sigma}, b_{\sigma})$, where $a_{\sigma}$ and $b_{\sigma}$ are fixed rate and scale parameters for the inverse-Gamma distribution. As in \eqref{model:lin_sys_dis}, the prior distribution for $\theta$ is absorbed into \eqref{model:lin_sys_con} and the posterior distribution for $\{\theta,\sigma^2\}$ for any fixed set $\{\delta_{\beta}, \delta_z, \xi, \phi\}$ is in the Normal-IG family. We exploit these distributions to devise stacked inference for $\beta_j(t)$ and $z(\gamma(t),t)$. 

\section{Prediction via stacking}\label{section:stacking}

We exploit the analytical closed forms for the posterior distributions and carry out inference using Bayesian stacking~\citep{le2017bayes,zhang2023exact}. In both discrete-time and continuous-time trajectory models, we are able to obtain closed-form posterior distributions if we fix some hyperparameters in the spatial-temporal covariance structures. We consider a collection of $G$ models, $\{\mathcal{M}_1,\ldots,\mathcal{M}_G\}$, where each $\mathcal{M}_g$ is specified by fixing a set of parameters so that the posterior distribution, given dataset $\mD$, $p_g(\cdot\given \mD )$, is in closed form. 

Specifically for the discrete time model in \eqref{model:lin_sys_dis}, the posterior distribution for $\mM_g$ is
\begin{equation}\label{eq: posterior_density_dlm}
    p_g\left(\theta,\sigma^2 \given \mD, \delta^2_g, \phi_g\right) = IG\left(\sigma^2 \given a_{\sigma}^*, b_{\sigma}^*\right)\times N\left(\theta \given m, \sigma^2 \Sigma\right)\;, 
\end{equation}
where $\delta^2_g = \{\delta^2_{\beta,g} \delta^2_{z,g}\}$ and $\phi_g$ are the fixed values of these parameters for $\mathcal{M}_g$. The posterior predictive distribution $p_g(y_{T+1} ( \gamma(T+1) )\mid \mD)$ and one for the latent process $p_{g}(z_{T+1} \given \mD)$ are both $t$-distributions with degrees of freedom, mean and scale supplied in Section~\ref{sec: con-dis-theory}. Similarly, for the continuous time model in~\eqref{model:lin_sys_con}, the posterior distribution for $\mM_g$ is
\begin{equation}\label{eq: posterior_density_con}
    p_g\left(\theta,\sigma^2 \given \mD, \delta^2_g, \phi_g, \xi_g \right) = IG\left(\sigma^2 \given a_{\sigma}^*, b_{\sigma}^*\right)\times N\left(\theta \given m, \sigma^2 \Sigma\right)\;, 
\end{equation}
where $\delta^2_g = \{\delta^2_{\beta,g}, \delta^2_{z,g}\}$, $\phi_g = \{\phi_{1,g} ,\phi_{2,g}\}$ and $\xi_g$ are the fixed values of these parameters for $\mathcal{M}_g$. The posterior predictive distributions of $y( \gamma(t_0), t_0 )$ and $z(\gamma(t_0),t_0)$ at the new time point $t_0$ are calculated from t-distributions; see Section~\ref{sec: con-con-theory} for details.

\subsection{Predictive stacking of means}\label{sec: stacking_means}
We divide the dataset $\mD$ into training data $\mD_{train}$ and validation data $\mD_{valid}$. We denote the predictive random variable by $\tilde{y}_t(\gamma(t))$ and $\tilde{y}(\gamma(t),t)$ at any given $t$ for the discrete and continuous time settings, respectively. We calculate the posterior predictive mean $\mathbb{E}_g [\tilde{y}_t(\gamma(t)) \mid \mD_{train} ]$ for each time point $t$ in the validation dataset in the discrete model in~\eqref{model:lin_sys_dis}, where $\mathbb{E}_g[\cdot]$ is the expectation with respect to the predictive density $p_g(\tilde{y}_t(\gamma(t)) \mid \mD_{train})$. Specifically, if $\tilde{y}_t$ is the vector with elements $\tilde{y}_t(\gamma(t))$ and $\tilde{X}_t$ is the matrix with rows $x_t(\gamma(t))^{\T}$ for each $\gamma(t)\in \mD_{valid}$, then
\begin{equation}\label{eq: predictive_mean_discrete}
    \mathbb{E}_g [\tilde{y}_t \mid \mD_{train}] = \tilde{X}_{t}\hat{\beta}_{t-1} + C_{z0}^{\T}{C}_z^{-1} \hat{z}_{t-1},
\end{equation}
where $\hat{\beta}_{t-1}$, $C_{z0}$, $C_{z}$ and $\hat{z}_{t-1}$ are described in Proposition~\ref{prop:pred-dis} of Section~\ref{sec: con-dis-theory}. We write $ \mathbb{E}_g [\tilde{y}_t(\gamma(t)) \mid \mD_{train}]$ to denote the element corresponding to $\gamma(t) \in \mD_{valid}$ in \eqref{eq: predictive_mean_discrete}. Likewise, in the continuous time model in \eqref{model:lin_sys_con}, the posterior predictive mean is $\mathbb{E}_g [\tilde{y}(\gamma(t), t) \mid \mD_{train}]$ for each $t$ in the validation set with $\mathbb{E}_g[\cdot]$ defined with respect to $p_g(\tilde{y}(\gamma(t),t) \mid \mD_{train})$, which is available in closed form as 
\begin{equation}\label{eq: predictive_mean_continuous}
    \mathbb{E}_g [\tilde{y}(\gamma(t),t) \mid \mD_{train}] = \sum_{j=1}^p{x}_{j,0} C_{\beta0}^{\T}{C}_{\beta}^{-1}\hat{\beta}_{j} + C_{z0}^{\T}{C}_z^{-1}\hat{z},
\end{equation}
where $\hat{\beta}_{j}$, $C_{\beta 0}$, $C_{\beta}$ and $\hat{z}$ are defined in Proposition~\ref{prop:pred-con} of Section~\ref{sec: con-con-theory}.

\begin{algorithm}[t]\caption{: Predictive stacking of means \label{algo:point} }
    \begin{algorithmic}[1]
    \State \textbf{Input:} $\mM_1, \ldots, \mM_G$, and dataset $\mD$. 
    \Statex \hspace{-1em}\textbf{Step 1: Calculate Predictions for Each Model} 
    \State Split the dataset into $\{\mD_{train}, \mD_{valid} \}$.
        \NoDo
        \For{$g = 1,2,\ldots,G$}
            \State Compute $\mathbb{E}_g [\tilde{y}_t \given \mD_{train}]$ or $\mathbb{E}_g [\tilde{y}(\gamma(t),t) \given \mD_{train}]$ from \eqref{eq: predictive_mean_discrete}~or~\eqref{eq: predictive_mean_continuous} for $t$ in $\mD_{valid}$. 
        \EndFor
    \Statex \hspace{-1em}\textbf{Step 2: Determine Weights}
    \State Determine the weights $\{\hat{a}_g\}_{g=1}^G$ to minimize the MSE, $\sum_{y\in \mD_{valid}} \left( y - \sum_{g=1}^G a_g \mathbb{E}_g \left[ \tilde{y} \mid \mD_{train}\right] \right)^2$, by the quadratic programming.
    \Statex \hspace{-1em}\textbf{Step 3: Compute Final Prediction}
    \State Compute the final prediction $\sum_{g=1}^G \hat{a}_g \mathbb{E}_g [\tilde{y} \given \mD_{train}]$.
    \end{algorithmic}
\end{algorithm}

Predictive stacking calculates the optimal weights to be used for model averaging. For stacking of means, we predict using $\sum_{g=1}^G a_g \mathbb{E}_g [\cdot \mid \mD_{train}]$, where $a_1,\ldots,a_G$ are the weights for model averaging selected from $\Delta = \{ \{a_g\}_{g=1}^G \mid \sum_{g=1}^Ga_g=1, a_g\ge 0 \}$, which yields a simplex of predictions on candidate models $\{\mM_g\}_{g=1}^G$. We determine the optimal weights $\{\hat{a}_g\}_{g=1}^G$ using the validation dataset as $\argmin_{a_1,\ldots,a_G} \sum_{{y}\in \mD_{valid}} \left(y - \sum_{g=1}^G a_g \mathbb{E}_g \left[\tilde{y} \mid \mD_{train}\right] \right)^2$, where the sum is over all values of the outcome in the validation dataset. This is a quadratic programming problem~\citep{Goldfarb1983, Boyd_Vandenberghe_2004}. The obtained weights are subsequently used to predict the outcomes using the stacked mean $\sum_{g=1}^G \hat{a}_g \mathbb{E}_g [\tilde{y} \given \mD]$, where $\tilde{y}$ corresponds to a specified $t_0$ in the sequence of time-points in the discrete-time case, while in continuous time $\tilde{y}$ represents the value $y(\gamma(t_0),t_0)$ for an arbitrary $t_0\in \mathbb{R}^{+}$. Algorithm~\ref{algo:point} summarizes these steps.

\subsection{Predictive stacking of distributions}\label{sec: stacking_dist}
Spatial energetics pursues predictive inference on trajectories entailing interpolation of the latent process at arbitrary points. This drives predictions for the outcomes. We achieve this by stacking the posterior predictive distributions for each $\mM_g$ using $\mD_{train}$, which is a multivariate t-distribution. Similar to the stacking of means, we consider the weights. Stacking maximizes the score function $S \left( \sum_{g=1}^G a_g p_g(\cdot \mid \mD_{train}) , q_t(\cdot\mid \mD_{train})\right)$ to obtain the weights, where $q_t$ is a posterior distribution with true underlying parameters. If we employ a logarithmic score, corresponding to the Kullback--Leibler divergence \citep{yao2018using, zhang2023exact}, the weights are obtained as $\argmax_{a_1,\ldots,a_G}\sum_{y\in \mD_{valid} } \log \left( \sum_{g=1}^G a_g p_g( y \mid \mD_{train}) \right)$, the logarithm acts on the pseudo-posterior probabilities, given the weighted models. Thus, we define the distributional prediction by maximizing the pseudo-log joint posterior probability. Note that $p_g(\cdot \mid \mD_{train})$ is a multivariate t-distribution, and hence, evaluating the posterior probability of the validation data is readily available. This optimization problem can be solved as a linearly constrained problem via an adaptive barrier algorithm~\citep{Lange2010}. Algorithm~\ref{algo:dist} presents the steps involved in stacking of predictive densities.

\begin{algorithm}[t]\caption{: Predictive stacking of distributions \label{algo:dist} }
    \begin{algorithmic}[1]
    \State \textbf{Input:} $\mM_1, \ldots, \mM_G$ and dataset $\mD$.
    \Statex \hspace{-1em} \textbf{Step 1: Calculate Predictive Distributions for Each Model on Each Fold}
    \State Split the dataset into $\{\mD_{train}, \mD_{valid} \}$.
        \NoDo
        \For{$g = 1,2,\ldots,G$}
            \State $p_g(\cdot\given \mD_{train}) \gets$ the posterior predictive distribution, by Proposition~\ref{prop:pred-dis} or~\ref{prop:pred-con}.
        \EndFor
    \Statex \hspace{-1em} \textbf{Step 2: Determine Weights}
    \State Determine the weights $\{\hat{a}_g\}_{g=1}^G$ to maximize $\sum_{y\in \mD_{valid} } \log \left( \sum_{g=1}^G a_g p_g(y\mid \mD_{train}) \right)$ through adaptive barrier method.
    \Statex \hspace{-1em} \textbf{Step 3: Compute Final Prediction}
    \State Compute the final predictive distribution $\sum_{g=1}^G \hat{a}_g p_g(\tilde{y}\given \mD)$.
    \end{algorithmic}
\end{algorithm}

The $G$ candidate models in Algorithms~\ref{algo:point}~and~\ref{algo:dist} are computed in parallel and computation of the weights is negligibly small compared to that of the posterior distribution. Further, optimization is supported by many packages in various statistical programming languages. In particular, for the subsequent illustrations, we employed the \citep[``stats'' and ``quadprg'' packages][]{stats,quadprog} in the \texttt{R} statistical computing environment. By contrast, MCMC demands a substantial number of iterations for convergence, and the issue is exacerbated with the larger values of $n$ and $T$.

\subsection{Reconstructing stacked posterior distributions}\label{sec: stacked_posterior}
Once the stacking weights are calculated from either Algorithm~\ref{algo:point}~or~\ref{algo:dist}, we use them to reconstruct the posterior distributions of interest as
\begin{equation}\label{eq: poterior density_stacked}
    p(\cdot \given \mD_{}) = \sum_{g=1}^G \hat{a}_g p_g(\cdot \given \mD_{})\;,
\end{equation}
where $\cdot$ represents the inferential quantity of interest. This embodies stacked inference for $\{\theta,\sigma^2\}$ in \eqref{model:lin_sys_dis}~and~\eqref{model:lin_sys_con}, predictions of the outcome $y_t(\gamma(t))$ at a future time point on a given trajectory or $y(\gamma(t),t)$ for any arbitrary time point on a trajectory, and inference for the latent process $z_t(\gamma(t))$ or $z(\gamma(t),t)$ in the discrete and continuous time settings, respectively. 

\section{Theoretical properties}\label{sec: theory}
\subsection{Distribution theory for Bayesian DLMs}\label{sec: bayesian_dlm_theory}
Fixing $\delta_{\beta},\delta_z,\phi$ yields closed-form posterior distributions for ${\theta}_t$ and $\sigma$, which aids stacking. We collect recursion equations used in calculating the posterior distribution for \eqref{model:dlm}.

\begin{proposition}\label{prop:post} Consider the model in \eqref{model:dlm}. Let $\mD_{t}$ denote all the data obtained until time $t$. Assume $\sigma^2\mid \mD_{\chi,t-1} \sim IG(n_{t-1}/2, n_{t-1}s_{t-1}/2)$ and ${\theta}_{t-1} \mid \sigma^2, \mD_{\chi,t-1} \sim N({m}_{t-1},\sigma^2W_{t-1})$. If $\delta_{\beta},\delta_z,\phi,G_t$ are fixed, the following distributional results hold, for $t \geq 1$,
\begin{align*}
    \sigma^2 \mid \mD_{t} &\sim IG\left( \frac{n_t}{2},\frac{n_t s_t}{2}  \right)\;;\;\quad {\theta}_{t} \mid \sigma^2, \mD_{t} \sim N({m}_{t},\sigma^2 W_{t}),
\end{align*}
where $n_t = n_{t-1}+n$, $n_ts_t = n_{t-1}s_{t-1} + ({y_t} - {f}_t )Q_t^{-1}({y_t} - {f}_t ), {f}_t = F_t G_t {m}_{t-1}, Q_t = F_t R_t F_t^{\T} + I_n, {m}_{t} = G_t{m}_{t-1} + R_t F_t^{\T} Q_t^{-1} ({y}_t-{f}_t),R_t = G_tW_{t-1}G_t^{\T} + S$ and $W_t = R_t - R_tF_t^{\T}Q_t^{-1}F_tR_t$.
The marginal posterior distribution of ${\theta}_{t}$ is $t_{n_t}({m}_{t},s_tW_{t})$.
\end{proposition}

Propositions \ref{prop:pred-s}~and~\ref{prop:pred-t} provide the spatial and temporal posterior predictive distributions.

\begin{proposition}\label{prop:pred-s}
Consider the setup for the model in \eqref{model:dlm} adapted for spatial data over $n$ locations $\chi = \{s_1,\ldots, s_n\}$. Let $\chi_0$ be a set of $n_0$ locations where we seek to predict $y_t(s)$ and $\tilde{X}_{0}$ be an $n_0\times p$ matrix of explanatory variables with rows $x_t^{\T}(s)$ for $s\in \chi_0$. If ${y}_{0}$ and ${z}_{0}$ denote the $n_0\times 1$ random variables corresponding to $y_t(s)$ and spatial effects $z_t(s)$ for all $s\in \chi_0$, then the posterior predictive distributions are
\begin{align*}
    {y}_{0} \mid {\theta}_{t}, {z}_{0} , \sigma^2, \mD_{t} &\sim N\left(\tilde{X}_{0} {\theta}_{t,(1:p)} + {z}_{0},\ \sigma^2 I_{n_0} +\sigma^2 \tilde{X}_{0} W_{t,(1:p,1:p)} \tilde{X}_{0}^{\T} \right),   \\
    {z}_{0} \mid {\theta}_{t}, \sigma^2, \mD_{t} &\sim N\left(C_0^{\T}{C}^{-1}{\theta}_{t,(p+1:p+n)},\ \sigma^2 (C_{00} - C_0^{\T}{C}^{-1}C_0) \right),
\end{align*}
where ${m}_{t,(1:p)}, {m}_{t,(p+1:p+n)}$ are the first $p$ elements and the remaining elements of ${m}_{t}$, $W_{t,(1:p,1:p)}$ is the top-left $p\times p$ square of $W_t$, $C = (\delta_z^2K_{\phi}(s,s'))_{s,s'\in \chi}$, $C_0=(\delta_z^2K_{\phi}(s,s_0))_{s\in\chi,s_0\in\chi_0}$ and $C_{00} = (\delta_z^2K_{\phi}(s_0,s'_0))_{s_0,s_0'\in \chi_0}$. Combined with Proposition~\ref{prop:post}, the marginal predictive distribution for ${y}_{t0}$ is $t_{n_t}( \tilde{X}_{0} {m}_{t,(1:p)} + C_0^{\T}{C}^{-1}{m}_{t,(p+1:p+n)}, s_t( I_{n_0} + \tilde{X}_{0} W_{t,(1:p,1:p)} \tilde{X}_{0}^{\T}+ C_{00} - C_0^{\T}{C}^{-1}C_0) )$.
\end{proposition}

\begin{proposition}\label{prop:pred-t}
Consider the assumptions in Proposition~\ref{prop:post}. The one-step ahead forecast and corresponding predictive distributions for the state vector are ${y}_{t+1} \mid {\theta}_{t+1}, \sigma^2, \mD_{t} \sim N\left(F_{t+1}{\theta}_{t+1} ,\ \sigma^2 I_n\right)$ and ${\theta}_{t+1} \mid \sigma^2, \mD_{t} \sim N\left(G_{t+1}{m}_{t},\ \sigma^2 (G_{t+1} W_{t} G_{t+1}^{\T} + S_{t+1}) \right)$, respectively.
The predictive distribution for ${y}_{t+1}$ is $t_{n_t}(F_{t+1}G_{t+1} {m}_{t}, s_t (I_n+ F_{t+1}( G_{t+1} W_{t} G_{t+1}^{\T} + S_{t+1})F_{t+1}^{\T}) )$. A general h-step ahead forecast can be obtained using recursive calculations.
\end{proposition}

\subsection{Distribution theory for discrete time trajectory model}\label{sec: con-dis-theory}
Fixing $\delta_{\beta}$, $\delta_z$, and $\phi$ produce accessible posterior distributions for the model in the trajectory regression~\eqref{model:condis}, which facilitates predictive stacking discussed in Section~\ref{section:stacking}. We present these posterior distributions below.
\begin{proposition}\label{prop: pos-dis}Posterior distribution of $({\theta},\sigma^2)$ in \eqref{model:condis}~and~\eqref{model:lin_sys_dis} is given by
\begin{align*}
    p({\theta},\sigma^2 \mid \mD) &= p(\sigma^2 \mid \mD)\times p({\theta} \mid \sigma^2,\ \mD) = IG(\sigma^2 \given a_{\sigma}^*,b_{\sigma}^* ) \times N(\theta \given {{m}}, \sigma^2 \Sigma),
\end{align*}
where $a_{\sigma}^* = a_{\sigma} + T/2,\ b_{\sigma}^* = b_{\sigma} + (Y-X{{m}}) S^{-1} (Y-X{m}),\ {m}=\Sigma X^{\T}S^{-1}Y ,\ \Sigma^{-1} = X^{\T} S^{-1} X$. The marginal posterior distribution of ${\theta}$ is $t_{2a_{\sigma}^*} \left({m},\ (b_{\sigma}^*/a_{\sigma}^*)\Sigma \right)$.
\end{proposition}

The posterior predictive distributions for future points on a trajectory is also available.

\begin{proposition}\label{prop:pred-dis}
Consider the setup for \eqref{model:lin_sys_dis} and Proposition~\ref{prop: pos-dis}. Let $\Gamma_0 = \{\gamma(1),\ldots,\gamma(T+1)\}$ be an enumeration of spatial locations and $\tilde{\Gamma}_0= \{\tilde{\gamma}_1,\ldots,\tilde{\gamma}_{n_0} \} \subseteq \Gamma_0$ be the set of $n_0$ distinct locations. Given $\mD$ up to time $T$, the predictive distributions at $T+1$ in \eqref{model:condis}~and~\eqref{model:lin_sys_dis} are
\begin{align*}
    y_{T+1} ( \gamma(T+1) )\mid {\beta}_{T+1},{z}_{T+1}, \sigma^2, \mD &\sim N\left({x}_{T+1}(\gamma(T+1))^{\T} {\beta}_{T+1} + \tilde{B} z_{T+1},\sigma^2 \right)\\
    {z}_{T+1} \mid {\theta}, \sigma^2, \mD &\sim N\left( C_{z0}^{\T}{C}_z^{-1}{z}_T , \sigma^2(C_{z00} - C_{z0}^{\T}{C}_z^{-1}C_{z0})\right),\\
    {\beta}_{T+1} \mid {\theta}, \sigma^2, \mD &\sim N\left( {\beta}_T ,\sigma^2\delta_{\beta}^2 W_p \right),
\end{align*}
where $C_z = (\delta_z^2K_{\phi}(\gamma,\gamma'))_{\gamma,\gamma'\in \tilde{\Gamma}}$, $C_{z0}=(\delta_z^2K_{\phi}(\gamma,\gamma_0))_{\gamma\in\tilde{\Gamma},\gamma_0\in\tilde{\Gamma}_0}$, $C_{z00} = (\delta_z^2K_{\phi}(\gamma_0,\gamma'_0))_{\gamma_0,\gamma_0'\in \tilde{\Gamma}_0}$ and $\tilde{B} = (b(\gamma(T+1),\tilde{\gamma}_j))_{j}$ is the $1 \times n_0$, constructed by the kernel $b(\cdot,\cdot)$ defined in Section~\ref{sec:con-dis}. The marginal predictive distribution for $y_{T+1} ( \gamma(T+1) )$ is $t_{2a_{\sigma}^*} ( {x}_{T+1}(\gamma(T+1))^{\T} \hat{\beta}_{T} + \tilde{B}C_{z0}^{\T}{C}_z^{-1}\hat{z}_T , (b_{\sigma}^*/a_{\sigma}^*) (1 + \delta_{\beta}^2 {x}_{T+1}(\gamma(T+1))^{\T}  W_p {x}_{T+1}(\gamma(T+1)) +\tilde{B}( C_{z00} - C_{z0}^{\T}{C}_z^{-1}C_{z0})\tilde{B}^{\T} ))$, where $\hat{\beta}_{T}$ and $\hat{z}_T$ are the posterior means calculated as $m$ in Proposition~\ref{prop: pos-dis}.
\end{proposition}

\subsection{Distribution theory for continuous time trajectory model}\label{sec: con-con-theory}
To exploit familiar results concerning \eqref{model:reg_con} that are used for stacking, as discussed in Section~\ref{section:stacking}, we fix $\delta_{\beta},\delta_z,\xi$ and $\phi$. The analytical posterior distributions are described below.
\begin{proposition} \label{prop: pos-cos} Posterior distribution of $({\theta},\sigma^2)$ in~\eqref{model:reg_con}--\eqref{model:lin_sys_con} is given by
\begin{equation*}
    p({\theta},\sigma^2 \mid \mD) = p(\sigma^2 \mid \mD)\times p({\theta} \mid \sigma^2,\ \mD) = IG(a_{\sigma}^*,b_{\sigma}^* ) \times N({{m}} , \sigma^2 \Sigma),
\end{equation*}
where $a_{\sigma}^* = a_{\sigma} + n/2,\ b_{\sigma}^* = b_{\sigma} + (Y-X{{m}}) S^{-1} (Y-X{m}),\ {m}=\Sigma X^{\T}S^{-1}Y ,\ \Sigma^{-1} = X^{\T} S^{-1} X$. The marginal posterior distribution of ${\theta}$ is $t_{2a_{\sigma}^*} \left({m}, (b_{\sigma}^*/a_{\sigma}^*)\Sigma \right)$.
\end{proposition}

The predictive distributions for new points on a trajectory are obtained as follows.
\begin{proposition} \label{prop:pred-con}
Consider the setup leading to \eqref{model:lin_sys_con} and Proposition~\ref{prop: pos-cos}. Let $(\Gamma_0, \mT_0)$ be the collection of $n_0$ new space-time points on a trajectory, $x_{j,0}$ be an $n_0\times n_0$ explanatory matrix at $(\Gamma_0, \mT_0)$ for $j=1,\ldots,p$, ${y}_{0}$ and ${z}_{0}$ be $n_0\times 1$ random variables corresponding to $y(\gamma(t),t)$ and $z(\gamma(t),t)$ for $(\gamma(t),t) \in \Gamma_0\times \mT_0$, and each $\beta_{j,0}$ is $n_0\times 1$ comprising $\beta_j(t)$ for $t\in \mT_0$. Then, the posterior predictive distributions are
\begin{align*}
    {y}_{0} \mid {\beta}_{1,0}, \ldots,{\beta}_{p,0} , {z}_0, \sigma^2, \mD &\sim N\left( \sum_{j=1}^p{x}_{j,0} {\beta}_{j,0} + {z}_{0},\ \sigma^2 I_{n_0} \right),  \\
    {z}_{0} \mid {\theta},\sigma^2, \mD &\sim N\left(C_{z0}^{\T}{C}_z^{-1}{\theta}_{np+1:np+n},\ \sigma^2(C_{z00} - C_{z0}^{\T}{C}_z^{-1}C_{z0}) \right),\\
    {\beta}_{j,0} \mid {\theta},\sigma^2, \mD &\sim N\left(C_{\beta0}^{\T}{C}_{\beta}^{-1}{\theta}_{(j-1)+1:nj},\ \sigma^2(C_{\beta00} - C_{\beta0}^{\T}{C}_{\beta}^{-1}C_{\beta0})\right),\quad j=1,\ldots,p,
\end{align*} 
where $C_{z} = (\delta_z^2K_{\phi}((\gamma(t),t),(\gamma(t'),t')))_{t,t'\in \mT}$, $C_{z 0}=(\delta_z^2K_{\phi}((\gamma(t),t),(\gamma(t_0),t_0)))_{t\in\mT,t_0\in\mT_0}$, $C_{z 00}=(\delta_z^2K_{\phi}((\gamma(t_0),t_0),(\gamma(t_0'),t_0')))_{t_0,t_0'\in\mT_0}$, $C_{\beta} = (\delta_{\beta}^2C_{\xi}(t,t'))_{t,t'\in \mT}$, $C_{{\beta}0}=(\delta_{\beta}^2C_{\xi}(t,t_0))_{t\in\mT,t_0\in\mT_0}$ and $C_{{\beta}00} = (\delta_{\beta}^2C_{\xi}(t_0,t_0'))_{t_0,t_0'\in \mT_0}$. The predictive distribution for ${y}_{0}$ is $t_{2a_{\sigma}^*}( \sum_{j=1}^p{x}_{j,0} C_{\beta0}^{\T}{C}_{\beta}^{-1} \hat{\beta}_{j} +C_{z0}^{\T}{C}_z^{-1}\hat{z}
,(b_{\sigma}^*/a_{\sigma}^*)(I_{n_0} + C_{z00} - C_{z0}^{\T}{C}_z^{-1}C_{z0} + \sum_{j=1}^p{x}_{j,0} (C_{\beta00} - C_{\beta0}^{\T}{C}_{\beta}^{-1}C_{\beta0}) {x}_{j,0}^{\T} ) )$, where $\hat{z}$ and $\hat{\beta}_{j}$ for $j=1,\ldots,p$ are the posterior means calculated as $m$ in Proposition~\ref{prop: pos-cos}.
\end{proposition}

\subsection{Repeated sampling properties of posterior distributions}\label{sec: poscon}
Here, we investigate some theoretical results for the state-space setting. At the outset, it is worth recognizing that classical inference is rather limited because trajectory data, by definition, do not admit replicates at a single time point. Theoretical accessibility requires multiple, say $n$, spatial locations at each time point. The relevant setting here is the second adaptation of \eqref{model:dlm} discussed in Section~\ref{sec: bayesian_dlm} with multiple spatial locations at each epoch. 

For convenient notation, let $y_t(s)$ and $z_t(s)$ denote the response and the spatial process, respectively, where $s$ is a generic spatial location in $\mathbb{R}^d$. We assume $n$ spatial replicates $\chi_n = \{s_1,\ldots,s_n\}$ at each $t$ and consider \eqref{model:dlm} without the trend, i.e., $\beta_t=0$. Let $\mD_{\chi_n,t}$ now be the entire dataset until time $t$ with spatial replicates in $\chi_n$. Hence, 
\begin{equation}\label{model: dlm_no_trend}
     y_t =  z_t + \eta_t, \quad \eta_t \overset{\mathrm{i.i.d.}}{\sim} N(0,\sigma^2I_n);\quad 
     z_t = \alpha z_{t-1} + \eta_{\theta,t}, \quad \eta_{\theta,t} \overset{\mathrm{i.i.d.}}{\sim} GP(0,\sigma^2\delta_z^2K_{\phi}(\cdot,\cdot)),
\end{equation}
where $y_t$ and $z_t$ are each $n\times 1$ with elements $y_t(s_i)$ and $z_t(s_i)$, respectively, $\alpha$ is a fixed real number, and $K_{\phi}(s_i,s_j) =  \frac{2^{1-\nu}}{\Gamma(\nu)} \left( \frac{\|{s}_i-{s}_j\|}{\phi} \right)^{\nu} \mK _{\nu}\left( \frac{\|{s}_i-{s}_j\|}{\phi} \right)$ is the Mat\'ern kernel \citep{stein1999interpolation} defined for any pair of spatial locations ${s}_i$ and ${s}_j$ in a bounded region $\mS\subset \R^2$. The parameters $\phi>0$ and $\nu>0$ model spatial decay and smoothness, respectively, and $\mK_{\nu}$ is the modified Bessel function of the second kind of order $\nu$. Here, we fix $\nu$
and employ prior distributions $\sigma^2\sim IG(n_{\sigma}/2, n_{\sigma}s_{\sigma}/2)$ and ${z}_0 \mid \sigma^2 \sim N({m}_{0}, \sigma^2S_0)$.

Let $\{\sigma_*, \phi_*, \delta_{z*}\}$ be fixed values of the model parameters that are used to generate data from \eqref{model: dlm_no_trend}, $\phi'$ and $\delta_z'$ be fixed values, $\mathbb{P}_*$ be the probability law of $y_t(s)$ corresponding to $\{\sigma_*, \phi_*, \delta_{z*}\}$, and $\mathbb{P}'$ be the law corresponding to $(\sigma_*, \phi', \delta_z')$. We require the notion of equivalence of probability measures for subsequent results.

\begin{definition}[Equivalence of probability measures]
Let $P_1$ and $P_2$ be two probability measures on the measurable space $(\Omega, \mathcal{F})$. Measures $P_1$ and $P_2$ are termed \textit{equivalent}, denoted $P_1 \equiv P_2$, if they are absolutely continuous with respect to each other. That is, $P_1 \equiv P_2$, if $P_1(A) = 0 \Leftrightarrow P_2(A) = 0$ for any $A \in \mathcal{F}$.
\end{definition}

\begin{lemma}[]\label{lem:equivalence-t}
For any $\phi'>0$, there exists $\delta_z'$ such that $\mathbb{P}'\equiv\mathbb{P}_*$.
\end{lemma}

Lemma~\ref{lem:equivalence-t} implies that if the parameters are fixed at values different from the true (data generating) parameters, then the incorrectly specified model is equivalent to that with the true parameters with regard to the distribution of $y$. This extends Theorem~2.1 in \cite{tang2021identifiability}. Based on this, the following result on the error variance is available.

\begin{theorem}\label{thm:var}
Assume that the fixed parameters are $\phi'$ and $\delta_z'$, satisfying $\mathbb{P}'\equiv\mathbb{P}_*$, and that
\begin{equation}
\label{eq:space}
\max_{s \in \mS} \min_{1 \le i \le n} |s - s_i| \asymp n^{-\frac{1}{d}},
\end{equation}
where $a_n\asymp b_n$ means $a_n$ is bounded both above and below by $b_n$ asymptotically.
If we set $m_0 = 0$, $S_0 = K_{\phi'}(\chi)$, $n_{\sigma} < \infty$, $s_{\sigma} < \infty$ and $\alpha<\infty$ in the Mat\'ern model~\eqref{model: dlm_no_trend}, then the posterior distribution of $\sigma^2$ converges, as $n \to \infty$, to the degenerate distribution at $\sigma^2_{\ast}$, i.e.,
\begin{equation*}
    p(\sigma^2\mid \mD_{\chi_n,t} ) \rightsquigarrow  \delta(\sigma^2_*),\quad  \mathrm{ \mathbb{P}_*-a.s. }
\end{equation*}
where $\rightsquigarrow$ denotes weak convergence of $p(\sigma^2 \mid \mD_{\chi_n,t} )$ and $\delta(x)$ is Dirac measure at $x$.
\end{theorem}

If the spatial locations are not overly concentrated within $\mS$, then, as the number of replicates increases, the posterior distribution of $\sigma^2$ degenerates to a point-mass distribution at the true parameter value \citep[formally referred to as posterior strong consistency,][]{ghosal2017fundamentals}. See the Appendix for details on the assumption in Theorem~\ref{thm:var}.

Turning to prediction at a new point $s_0\in \mS$, let $\tilde{z}_t(s_0)$ and $\tilde{y}_t(s_0)$ be predictive random variables at any given time $t$ and $Z_{tn}(s_0)$ be a random variable with density $p(\tilde{z}_t(s_0)\mid \mD_{\chi_n,t})$ and $Y_{tn}(s_0)$ be a random variable with density $p(\tilde{y}_t(s_0)\mid \mD_{\chi_n,t})$. Let $\mathbb{E}_*[\cdot]$ denote the expectation with respect to $\mathbb{P}_*$. The prediction errors $\mathbb{E}_* [(Z_{tn}(s_0)-\tilde{z}_t(s_0))^2]$ for the latent process and $\mathbb{E}_* [(Y_{tn}(s_0)-\tilde{y}_t(s_0))^2]$ for the response are of interest.

\begin{theorem}\label{thm:ypred}
With assumptions in Theorem~\ref{thm:var}, the following results hold for $t=1,\ldots,T$:
    \begin{equation*}
      \mathbb{E}_* [(Z_{tn}(s_0)-\tilde{z}_t(s_0))^2] = E_{n,t}^A + E_{n,t}^B~\mbox{  and  }~ \mathbb{E}_* \left[(Y_{tn}(s_0)-\tilde{y}_t(s_0))^2\right] - E_{n,t}^A - E_{n,t}^B \to 2\sigma_*^2~\mathrm{as}~ n\to \infty,
    \end{equation*}
    where $E_{n,t}^A = \sigma_*^2\left(\delta_z'^2h_\phi(\chi_n,s_0) + g(\chi_n, s_0)\right) + o(1)$, $h_\phi(\chi_n,s_0)$ and $g_\phi(\chi_n,s_0)$ are functions defined as $h_\phi(\chi_n,s_0) = 1- K_{\phi}(\chi_n,s_0)^{\T} K_{\phi}^{-1}(\chi_n)K_{\phi}(\chi_n,s_0)$ and $g(\chi_n, s_0) = K_{\phi}(\chi_n,s_0)^{\T} K_{\phi}^{-1}(\chi_n) R_{t,n}(R_{t,n}  + I_n)^{-2}R_{t,n} K_{\phi}^{-1}(\chi_n)K_{\phi}(\chi_n,s_0)^{\T}$, $K_\phi(\chi_n,s_0) = \{K_{\phi}(s,s_0)\}_{s\in\chi_n}$ is $n\times 1$, and
    $E_{n,t}^B= \mathbb{E}_* [(\tilde{z}_t(s_0)- K_{\phi}(\chi_n,s_0)^{\T} K_{\phi}^{-1}(\chi_n) {m}_t )^2]$ with $R_{t,n} = \alpha ^2 W_{t-1} + \delta_z^{'2} K_{\phi'} (\chi_n)$, $W_t = R_{t,n} - R_{t,n}^{\T}Q_{t,n}^{-1}R_{t,n}$ and $Q_{t,n} =  R_{t,n}  + I_n$.
\end{theorem}

The initial finding relates to estimates of spatial-temporal effects. The estimation error at time $t$ is decomposed by $E_{n,t}^A$ and $E_{n,t}^B$, where $E_{n,t}^A$ is the variance term with its asymptotic form expressible without expectations, and $E_{n,t}^B$ is a bias term representing the difference between a true random variable and linear predictor by filtering $m_t$. The subsequent result indicates that the prediction error is articulated in relation to the measurement error scale and the two terms. As the number of measurement points increases while the observed area remains fixed, more points are close to the new points. Consequently, prediction accuracy is enhanced at new points, and $E_{n,t}^A$ and $E_{n,t}^B$ become small. An exploration of the convergence of $E_{n,t}^B$ in a limited scenario is presented in the Appendix. Since $E_{n,t}^A$ is challenging to analyze theoretically, we present numerical experiments on the decay of $E_{n,t}^A$ in the Appendix. We develop the following result concerning stacking.

\begin{theorem}\label{thm:stack} 
Under the assumptions in Theorem~\ref{thm:var}, if $E_{n,t}^B\to 0$ as $n\to \infty$, then
    \begin{equation*}
        \mathbb{E}_* \left[\left(\tilde{y}_t(s_0) - \sum_{g=1}^G a_g \mathbb{E}_g[ \tilde{y}_t(s_0)\mid \mD_{\chi_n,t} ] \right)^2\right] \to \sigma_*^2 , \quad  t=1,\ldots,T,
    \end{equation*}
where $\{a_g\}_{g=1}^G$ satisfies $\sum_{g=1}^G a_g=1$.
\end{theorem}
This theorem validates stacking as an inferential procedure. 
This is an asymptotic result so for finite samples the estimator of the process exhibits bias which can diminish predictive accuracy. Model averaging enhances stability and, in the context of statistical learning, the convex hull increases the flexibility of the model (or reduces the training error) without increasing its Rademacher complexity, i.e., the generalization gap. \citep[See Chapter~6 in][for details]{mohri2018foundations}. Hence, stacked inference delivers better predictive performance than a single model \citep[see][for further theoretical results]{vanderLaan2007}.

\section{Simulation}\label{section:simulation}
We illustrate implementation and inferential effectiveness of our proposed methods with numerical experiments. Section~\ref{sec:infill} explores in-fill prediction on a continuous trajectory, while Section \ref{sec:sim_proposed} illustrates the performance of our proposed methods. 

\begin{figure}[t]
\centering
\includegraphics[width=0.85\textwidth]{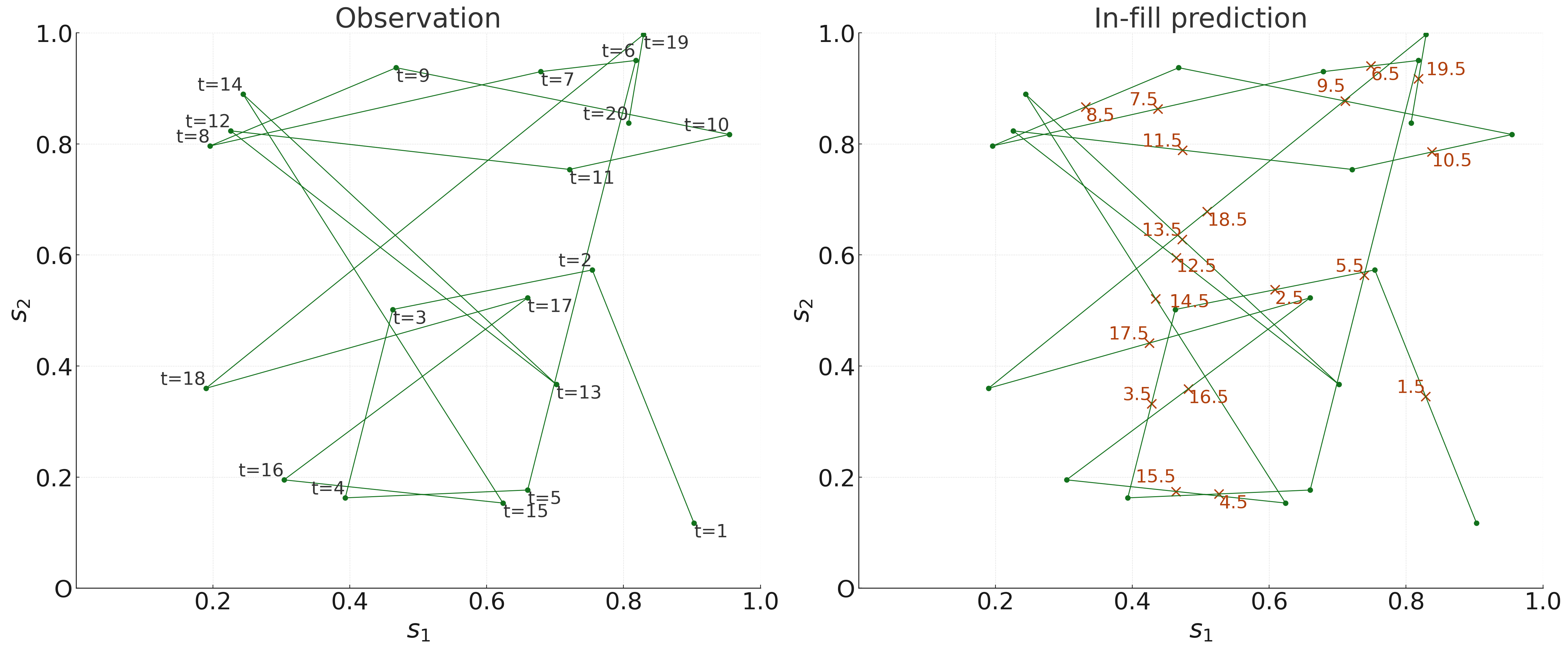}
\caption{Left: observation points and trajectory (green line). Right: unobserved points (red cross marks) on the trajectory, where forecasts are made. \label{fig:trajectory-continuous} }
\end{figure}

\subsection{In-fill prediction} \label{sec:infill}

We consider a single trajectory defined by $\gamma:\mT\to \mathbb{R}^2$, where $\gamma(t)$ over time interval $\mT$. Because data are observed at discrete points on this trajectory, the accuracy of in-fill prediction is expected to improve as the number of discrete observed points increases. We consider interpolation of the outcome over a space-time trajectory composed of line segments. As an example of in-fill prediction, the conceptual diagram in Figure~\ref{fig:trajectory-continuous} depicts $T=20$ observations in the interval $\mT= [1,20]$. The left panel displays the locations at each time point with the line segments comprising the trajectory shown by a green line. The right panel illustrates the space-time coordinates where we seek to predict the outcome.

We generate data along a trajectory comprising a fixed domain in accordance with \eqref{model:reg_con}--\eqref{model:t_process}. We generate $300$ points on the trajectory over $t\in \mT=[1,300]$ using $\gamma(t) = \gamma(t-1) + \mathcal{WN}(0,1)$, where $\mathcal{WN}(0,1)$ denotes white noise with zero mean and unit variance. We randomly include $n$ space-time coordinates over the trajectory and generate $n$ values of $z(\gamma(t),t)$ from the Gaussian process with $0$ mean and covariance kernel as \eqref{model:st_process}. The parameters defining the spatial-temporal process $z(\gamma(t),t)$ in \eqref{model:st_process} are set to $\phi_1=1/2$ and $\phi_2=1/2$. We then generate $y(\gamma(t),t)$ from \eqref{model:reg_con} using elements of $x(\gamma(t),t)$ generated from $N(0,4)$ and $\beta_j(t)$ using a zero-centered Gaussian process with covariance kernel in \eqref{model:t_process} specified by $\xi=1/2$. Additionally, we set $\sigma=1$, $\delta_{\beta}=1$, and $\delta_z=1$.

For our experiments, we include $n=20, 40,\ldots, 200$ points for training the data. Based on these observations, we evaluate performance of interpolation over trajectories by comparing the posterior predictive means of $y(\gamma(t),t)$ and $z(\gamma(t),t)$ over $100$ randomly selected points on the trajectory that were excluded from the training data. For predictive stacking, a set of candidate parameters is specified with $\phi_i \in \{1, 1/5\}$ for $i=1,2$ in \eqref{eq: stkernel}, $\xi\in \{1,1/5\}$ in \eqref{model:t_process}, and $\{3,1/3\}$ for both $\delta_{\beta}$ and $\delta_z$ in \eqref{model:st_process}~and~\eqref{model:t_process}. We employ $20$-fold cross-validation to obtain the stacking weights described in Sections~\ref{sec: stacking_means}~and~\ref{sec: stacking_dist} and use posterior estimates drawn from \eqref{eq: poterior density_stacked}. Our subsequent posterior summaries refer to \eqref{eq: poterior density_stacked}. Furthermore, we implement $M_0$, an oracle method with the true parameters assigned and Bayesian model averaging (BMA, \cite{Hoeting1999}) with a uniform prior on candidate models, which yields a weighted average of multivariate t-distributions; see the Appendix for details on BMA. As measures of performance, we adopted three metrics: mean squared prediction error (MSPE) and mean squared error for $z$ (MSE$z$) for stacking of means, and mean log predictive density (MLPD) for predictive stacking of distributions.

\begin{figure}[t]
    \centering
    \includegraphics[width=\textwidth]{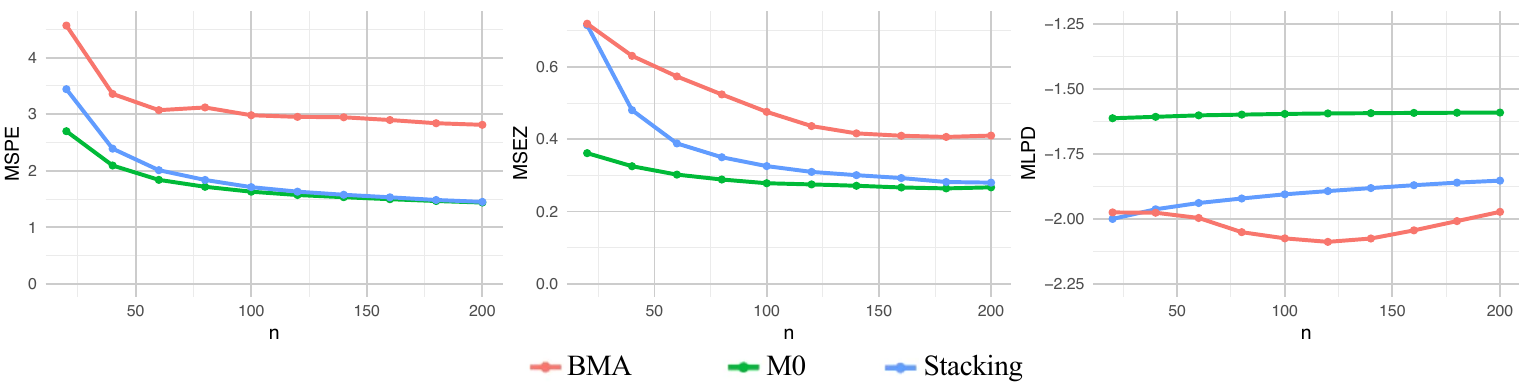}
    \caption{Predictive performances of stacking, BMA and $M_0$ (method with the oracle parameters) when $n$ grows from $20$ to $200$. Left: MSPE (mean squared prediction error). Middle: MSE$z$ (mean squared error for $z$). Right: MLPD (mean log predictive density). Stacking outperforms BMA and approaches the oracle more rapidly with increasing $n$. \label{fig:infill-continuous} }
\end{figure}

Figure~\ref{fig:infill-continuous} illustrates the overall predictive error behavior, where we generated 50 different datasets and report the average of the aforementioned metrics over the different datasets. The left and center panels, which plot MSPE and MSE$z$, respectively, demonstrate that an increasing number of points in the training data ($n$) over the fixed domain (trajectory) enhances the precision of predictions for outcomes and spatial effects. Similarly, the right panel reveals improvement in predictive accuracy in terms of MLPD as $n$ increases. Notably, stacking significantly outperforms BMA as its metrics approach those for the oracle model more rapidly with increasing $n$. While results established in Theorems~\ref{thm:var}--\ref{thm:stack} apply to Euclidean domains for theoretical tractability, our empirical findings on non-Euclidean trajectories in this experiment still appear to be consistent with those theoretical results.

\subsection{Estimation of proposed models}\label{sec:sim_proposed}

We conduct simulation experiments to assess discrete- and continuous-time trajectory models, comparing their performance in terms of estimation errors and model fitting. First, we sample data using the discrete time trajectory in \eqref{model:condis}--\eqref{model:lin_sys_dis}. We generate $n=T=50$ and $70$ points for two experiments on a trajectory using the same random walk model for $\gamma(t)$ as in Section~\ref{sec:infill}. We let $p=2$ in \eqref{model:condis} and generate values for the elements of ${\beta}_0$ and ${z}_0$ from $N({0},4)$. With these fixed values, we sequentially generate ${\beta}_t$ and ${z}_t$ using the autoregressive model (Section~\ref{sec:con-dis}) with $\delta_{\beta} = 1$, $\delta_z = 1$, $\sigma = 1$ and $K_{\phi}$ taken as the Mat\'ern kernel introduced in Section~\ref{sec: poscon} with $\phi = 1/7$ and $\nu = 1$. Each element of ${x}_t(\gamma(t))$ is generated from $N(0,4)$ and fixed thereafter. Then, $y_t(\gamma(t))$ is generated using \eqref{model:condis}.

We generate $30$ different datasets and analyze each of them using the models in \eqref{model:lin_sys_dis}~and~\eqref{model:lin_sys_con}. For predictive stacking, we set candidate parameters in \eqref{model:lin_sys_dis} as $\phi\in \{1, 1/10\}$, $\nu\in \{3, 1/3\}$ and $\{5,1/5\}$ for both $\delta_{\beta}$ and $\delta_z$. For \eqref{model:lin_sys_con}, we consider $\phi_i\in \{3, 1/10\}$ for $i=1,2$, $\xi\in \{3, 1/10\}$, and $\{3,1/3\}$ for $\delta_{\beta}$ and $\delta_z$. We employ $20$-fold expanding window cross-validation and select the stacking weights described in Sections~\ref{sec: stacking_means}~and~\ref{sec: stacking_dist}. We also estimate a non-spatial dynamic linear model (NSDLM),
\begin{equation*}
    y_t = {x}_t^{\top} {\beta}_t + \eta_{1t} , \quad \eta_{1t}\overset{\mathrm{i.i.d.}}{\sim} N({0},\sigma^2); \qquad   {\beta}_t = {\beta}_{t-1}  + \eta_{2t} , \quad \eta_{2t}\overset{\mathrm{i.i.d.}}{\sim} N({0},W ).
\end{equation*}
We assign priors ${\beta}_0\sim N({0},10I_2)$, $W\sim IW(10,10I_2)$ and $\sigma\sim IG(1/2,1/2)$. This model solely captures the temporal structure of the data without accounting for spatial locations.

We use MSE to compare posterior predictive means for $y_{T+1}$ with underlying signals in \eqref{model:condis}. We use relative mean squared errors (rMSE) defined as $\sum_{t=1}^T (\hat{\omega}_t - \omega_t)^2/(\sum_{t=1}^T\omega_t^2)$, where $\hat{\omega}_t$ is the posterior mean and $\omega_t$ is the true value, to compare our inferential effectiveness, We apply rMSE to each element of $\beta_t$ and $z_t$, while we use $(\hat{\sigma}-\sigma)^2/\sigma^2$ for $\sigma$. To evaluate model fit, we report MLPD, the deviance information criterion \citep[DIC,][]{spiegelhalter2002bayesian} and the widely applicable information criterion \citep[WAIC,][]{watanabe2010asymptotic}.

\begin{table}[t]
\caption{Comparison of methods for data generated by discrete-time trajectory models. MSE$y$ denotes the mean squared error of the response variable $y$, rMSE$z$ is the relative MSE (rMSE) of the spatial-temporal effect $z$, rMSE$\beta_1$ and rMSE$\beta_2$ are the rMSE of the two slope coefficients, respectively, and rMSE$\sigma$ is the rMSE of the $\sigma$. 
\label{tab:sim_comparison1}}
\centering
\begin{tabular}{lccccccc}
\toprule
& \multicolumn{3}{c}{\textbf{n=50}} & \multicolumn{3}{c}{\textbf{n=70}} \\
\cmidrule(lr){2-4} \cmidrule(lr){5-7}
\textbf{Metrics} & \textbf{Discrete} & \textbf{Continuous} & \textbf{NSDLM} & \textbf{Discrete} & \textbf{Continuous} & \textbf{NSDLM} \\
\midrule
MSE$y$ & 1.038 & 1.028 & 2.171 & 0.996 & 0.989 & 2.163 \\
rMSE$z$ & 0.761 & 0.908 & — & 0.766 & 0.915 & — \\
rMSE$\beta_1$ & 0.166 & 0.576 & 1.709 & 0.197 & 0.620 & 2.042 \\
rMSE$\beta_2$ & 0.201 & 0.648 & 1.698 & 0.142 & 0.570 & 1.651 \\
rMSE$\sigma$ & 0.564 & 0.583 & 0.295 & 0.208 & 0.227 & 0.298 \\
MLPD & -1.793 & -1.772 & -1.774 & -1.784 & -1.743 & -1.784 \\
DIC & 50.682 & 151.500 & 417.482 & 62.454 & 216.382& 586.966 \\
WAIC & 43.674 & 225.633 & 1909.787 & 80.012 & 269.553& 2058.120 \\
\bottomrule
\end{tabular}
\end{table}

Table~\ref{tab:sim_comparison1} compares the three models with regard to predictions for datasets generated by the discrete-time model. The entries show the average values of the metrics over $30$ datasets. The discrete and continuous trajectory models outperform NSDLM, as the latter fails to capture spatial structure. We note insignificant differences between the discrete and continuous time models with respect to MSE of $y$ (MSE$y$) and the rMSE of $z$ (rMSE$z$). This indicates that both methods are nearly comparable in their ability to evaluate the true signal and spatial-temporal effects in this dataset and indicates the competitiveness of the continuous model. However, the discrete model outperforms the continuous model in terms of rMSE of the regression coefficients in $\beta$ (rMSE$\beta_1$ and rMSE$\beta_2$), where the discrete model unsurprisingly captures its own data generating process better than the continuous model. For model evaluation, the discrete-time model excels in terms of DIC and WAIC, although the MLPDs are nearly equal across all the models.

\begin{figure}[t]
\centering
\includegraphics[width=0.8\textwidth]{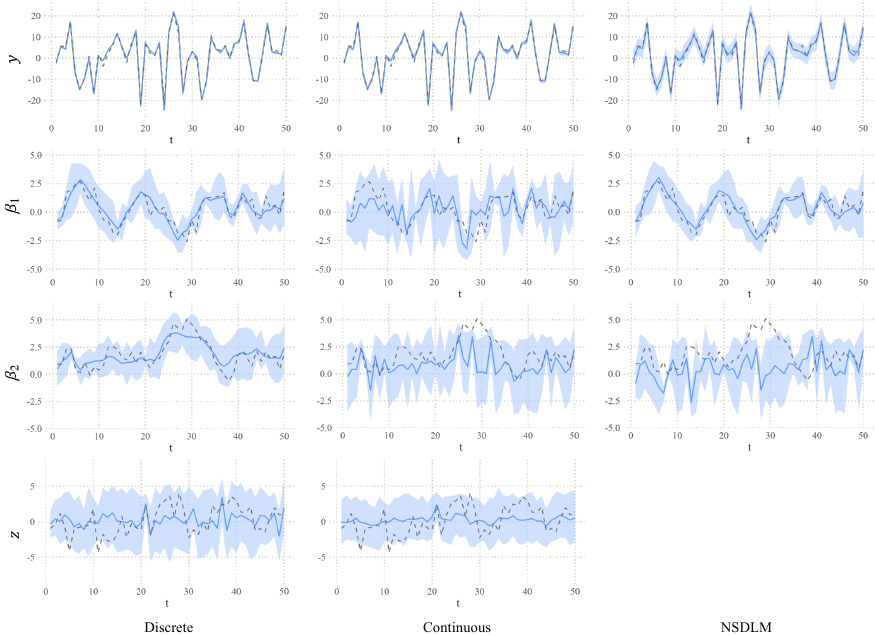}
\caption{Posterior bands of parameters obtained by discrete time, continuous time and NSDLM models for one representative dataset generated from a discrete time trajectory model with $n=50$. Each panel represents the posterior summaries for $y$, $z$ and $\beta$, with solid blue lines representing posterior means, and the shaded blue bands representing the 95\% credible intervals, and dashed lines representing true values.}\label{fig:sim1}
\end{figure}

Figure~\ref{fig:sim1} presents the posterior means (solid curve) for $y_t(\gamma(t))$, $z_t(\gamma(t))$ and two regression coefficients in each $\beta_t$ as a function of time for one representative dataset with $n=50$ in the simulation experiment. Also shown are the true values generating the data in the form of a dashed line. The posterior bands from the discrete model are effective in containing the dashed line (truth) with the continuous model and NSDLM still performing fairly well but missing the truth for the second coefficient in $\beta$.

Next, we generate $30$ datasets from the continuous-time trajectory model described in \eqref{model:reg_con}--\eqref{model:t_process}. We generate $n=T=50$ and $70$ points on the trajectory by the same random walk as in Section~\ref{sec:infill}. The true parameters for spatial-temporal process $z(\gamma(t),t)$ in \eqref{model:st_process} are $\phi_1 = 1/2, \phi_2 = 1/2$ and $\delta_z = 1$. Then, we produce $y$ from \eqref{model:reg_con} using $\sigma = 1$, each $\beta_j$ generated from \eqref{model:t_process} with $\xi=1/2$, $\delta_{\beta} = 1$ and elements of $x(\gamma(t),t)$ generated from $N(0,4)$.

We analyze each of the $30$ datasets using continuous time and discrete time trajectory models. For predictive stacking, we use $\phi\in \{2, 1/2\}$, $\nu\in\{2, 1/2\}$, $\delta_{\beta}\in\{1/2,1/10\}$ and $\delta_z\in\{1/2,1/10\}$ for the discrete model in \eqref{model:lin_sys_dis}  . For the continuous model in \eqref{model:lin_sys_con}, we set $\phi_i\in\{1, 1/4 \}$ for $i=1,2$, $\xi\in \{1,1/4\}$, and $\delta_{\beta}$ and $\delta_{z}$ in $\{3,1/3\}$. We employ $20$-fold expanding window cross-validation to determine the stacking weights as earlier.

\begin{table}[t]
\caption{Comparison of methods for data generated by continuous-time trajectory models. MSE$y$ represents the mean squared error of the denoised response variable $y$, rMSE$z$ the relative MSE (rMSE) of the spatial-temporal effect $z$, rMSE$\beta_1$ and rMSE$\beta_2$ the rMSE of the coefficient $\beta_1$ and $\beta_2$, respectively, and rMSE$\sigma$ the rMSE of the $\sigma$. \label{tab:sim_comparison2}}
\centering
\begin{tabular}{lccccccc}
\toprule
& \multicolumn{3}{c}{\textbf{n=50}} & \multicolumn{3}{c}{\textbf{n=70}} \\
\cmidrule(lr){2-4} \cmidrule(lr){5-7}
\textbf{Metrics} & \textbf{Discrete} & \textbf{Continuous} & \textbf{NSDLM} & \textbf{Discrete} & \textbf{Continuous} & \textbf{NSDLM} \\
\midrule
MSE$y$ & 0.810 & 0.725 & 0.997 & 0.863 & 0.674 & 0.902 \\
rMSE$z$ & 0.508 & 0.311 & — & 0.434 & 0.264 & — \\
rMSE$\beta_1$ & 0.088 & 0.093 & 2.737 & 0.109 & 0.134 & 6.282 \\
rMSE$\beta_2$ & 0.078 & 0.150 & 3.041 & 0.103 & 0.075 & 1.914 \\
rMSE$\sigma$  & 0.346 & 0.498 & 0.043 & 0.480 & 0.506 & 0.042 \\
MLPD & -1.279 & -1.272 & -1.414 & -1.350 & -1.213 & -1.394 \\
DIC & 139.965 & 124.344 & 222.319 & 177.520 & 173.052 & 284.848 \\
WAIC & 132.214 & 148.231 & 304.411 & 167.744 & 213.590 & 385.777 \\
\bottomrule
\end{tabular}
\end{table}

The results of the predictions by the three methods are presented in Table~\ref{tab:sim_comparison2}. As in Table~\ref{tab:sim_comparison1}, the continuous and discrete time models outperform NSDLM. The accuracy of the response variable $y$ and spatial-temporal effect $z$ indicate that the continuous model performs better than the discrete model. Regarding the model fitting, again as seen in Table~\ref{tab:sim_comparison1}, we see that both the continuous and discrete time models excel over NSDLM.

\begin{figure}[t]
\centering
\includegraphics[width=0.8\textwidth]{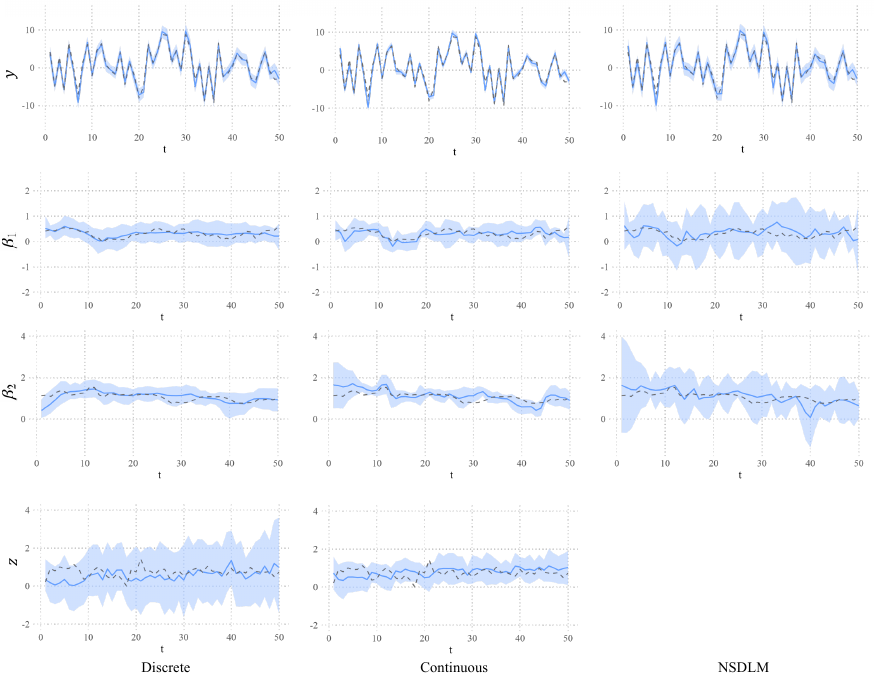}
\caption{Posterior distributions of parameters from discrete time, continuous time and NSDLM models for one representative dataset generated from a continuous time trajectory model. 
Solid blue lines are posterior means, shaded blue areas represent 95\% credible intervals, and dashed lines represent true values. }\label{fig:sim2}
\end{figure}

Figure~\ref{fig:sim2} is analogous to Figure~\ref{fig:sim1} for a representative dataset from the continuous time experiment showing posterior mean and 95\% intervals for $y(\gamma(t),t)$, $z(\gamma(t),t)$ and two slopes as a function of time. While all three methods seem comparable in capturing the truth (dashed line), the precision for the continuous model is higher than the other two models with NSDLM clearly having the widest bands for the regression slopes.  

\begin{table}[t]
\caption{Posterior summary of $\sigma^2$ of each method on each scenario when $n=50$. The true value of $\sigma^2$ is $1$. Both $\delta_z$ and $\delta_{\beta}$ are fixed to the true values.}
\centering
\begin{tabular}{@{}lcc@{}}
\toprule
\multirow{2}{*}{Method} & \multicolumn{2}{c}{Data Generating Process} \\
\cmidrule(l){2-3}
& Discrete & Continuous \\
\midrule
Discrete & $0.80$ (${0.54},{1.22}$) & $0.10$ ($0.06, 0.16$) \\
Continuous & $8.77$ ($2.04, 18.51$) & $1.02$ (${0.74}, {1.48}$) \\
\bottomrule
\end{tabular}
\label{table:sig}
\end{table}

Finally, we turn to the stacked posterior inference for $\sigma^2$ in Table~\ref{table:sig}. We present the posterior mean and the $95\%$ credible intervals obtained from one representative dataset generated by the discrete-time model and another from the continuous-time model when both $\delta_{\beta}$ and $\delta_z$ are fixed at their true values and $n=50$. The $2\times 2$ table presents the stacked estimates for $\sigma^2$ when each of these models is estimated from the two representative datasets. We find that the credible intervals from the discrete and continuous time models are able to capture the true value ($\sigma=1$) for the data generated from them, which seems to be consistent with the result in Theorem~\ref{thm:var}, but may not be able to capture the true value of $\sigma$ when they are not the data generating model. Furthermore, if $\delta_\beta$ and $\delta_z$ are not specified at their true values, inference for $\sigma^2$ suffers. This phenomenon has also been investigated by \cite{zhang2023exact} and is largely attributable to the fact that $\sigma^2$ is not consistently estimable in Gaussian process models \citep{zhang2004inconsistent, tang2021identifiability}.  

\section{Application}\label{section:application}
We apply \eqref{model:reg_con} to an actigraph dataset from the Physical Activity through Sustainable Transport Approaches in Los Angeles (PASTA-LA) described in \cite{alaimo2023bayesian}. Actigraph data are collected through wearable devices including sensors and a smartphone application, providing high-resolution and repeatable measurements for monitoring human activity. There is a growing body of research on relationships between \textit{Magnitude of Acceleration} (MAG) \citep{ishikawa2008physical, lyden2014method, staudenmayer2015methods, bai2018two} and physical activity measures such as \textit{energy expenditure} (EE) \citep{crouter2006novel, freedson2012assessment, taraldsen2012physical}. We consider instantaneous body vector MAG as the primary endpoint \citep{doherty2017large}. 

We seek to estimate an individual's MAG along a traversed path after accounting for the impact of explanatory variables on metabolism. These estimates from the continuous time model, which represent statistical learning of an individual's physical activity profile on any given day, are then used to predict the subject's metabolic activity on any arbitrary trajectory and comprise a personalized recommendation system for the subject. The models we have developed here can yield more effective pathways and environments to enhance physical activity levels of subjects and lead to overall improvements in metabolic levels. 

\begin{figure}[t]    
\centering
\includegraphics[width=0.9\textwidth]{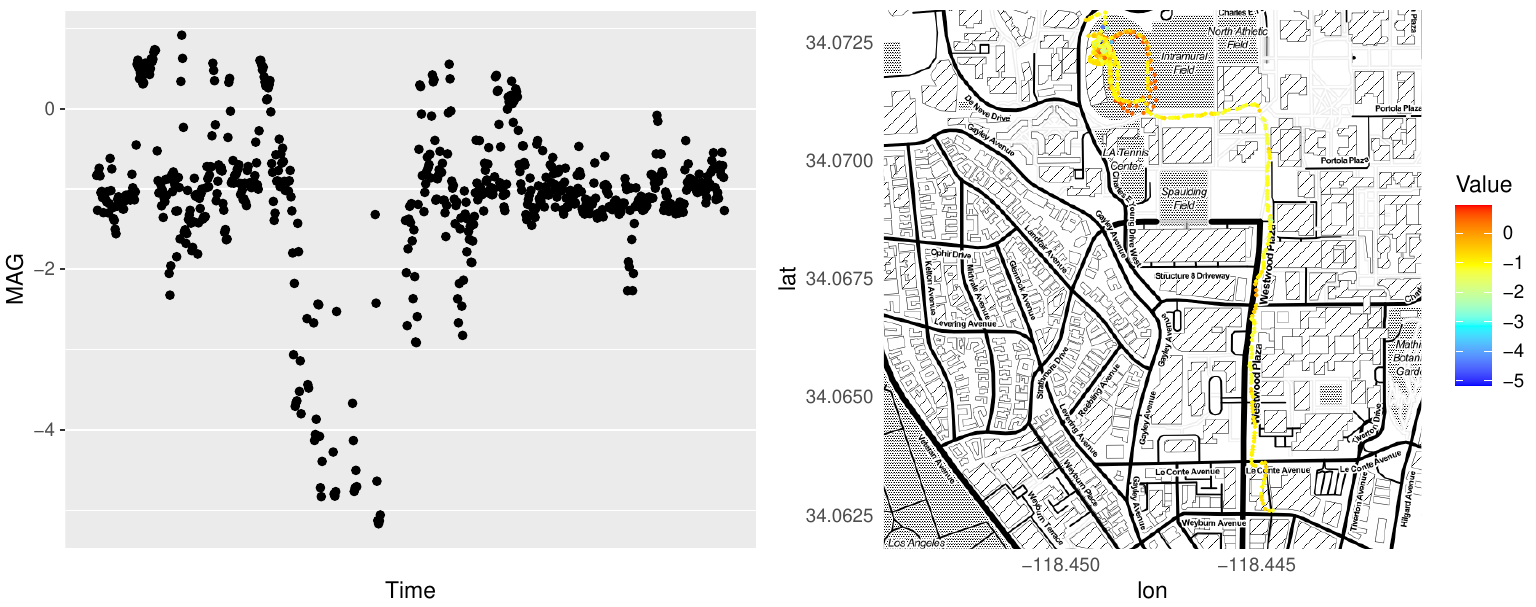}
\caption{$650$ time-series points of log-transformed MAG (left), and spatial movement of the person with log-transformed MAG values (right).
\label{fig:MAG}}
\end{figure}

Focusing on one individual, which is often the goal in personalized health science research seeking data driven recommendation systems for metabolic activity, we consider MAGs observed at $650$ unevenly spaced time points. The left panel of Figure~\ref{fig:MAG} displays the MAG values and the right panel plots the MAGs over the path traversed by this individual. We fit the model in \eqref{model:reg_con} using ``Slope'', representing the gradient at the spatial location, and ``NDVI'', representing normalized vegetation index, which is a measure of greenness at the location, as the two explanatory variables in $x(\gamma(t),t)$
. For predictive stacking, the candidates for the parameters $\phi_1, \phi_2, \xi$ are set $\{100, 1\}$ and those for $\delta_{\beta},\delta_z$ are set $\{20,1/5\}$. Of the $650$ observations, $130$ are randomly selected and used as test data, and the remaining $520$ are utilized as training data. We compute the posterior and predictive distributions for the continuous time trajectory model and Bayesian linear regression on the training data. For comparison, we apply a Bayesian linear regression model to this data using the ``brm'' function from the \citep[``brms'' package][]{R-brms} in \texttt{R} and adopt a (default) $t_3(0,2.5)$ prior for $\sigma$ and flat priors for coefficients of the explanatory variables.

\begin{table}[t]
    \caption{ MSPE and MLPD of the interpolations by Bayesian linear regression and the continuous time trajectory model with predictive stacking. \label{tab:app-interpolation} }
    \centering
    \begin{tabular}{ccc}
        \toprule
        \textbf{ } & \textbf{Bayesian linear regression} & \textbf{Continuous time trajectory model} \\
        \midrule
        MSPE & 0.787 & 0.108 \\
        MLPD & -1.312 & -0.883 \\
        \bottomrule
    \end{tabular}
\end{table}

Table~\ref{tab:app-interpolation} indicates the superiority of the proposed methods over Bayesian linear regression. There is significant difference in MSPE with the proposed model exhibiting considerably lower errors compared with Bayesian linear regression. The MLPD, indicating the goodness of the posterior distribution, unsurprisingly, demonstrates that the proposed model outperforms Bayesian linear regression by accounting for spatial-temporal structure in the data. The benefits of driving the inference using a spatial-temporal process are clear from this analysis. In particular, linear regression and ordinary least squares estimates. while informing about the impact of explanatory variables on MAGs, is severely limited in predictive capabilities as it fails to account for 
spatial-temporal information
.

Figure~\ref{fig:pred_acti} presents 3 maps to display spatial interpolation for actigraph data. The left panel plots the 130 raw observations over a path traversed by the subject under consideration. The values of the (transformed) MAG are calibrated using colors shaded from deep blue (lowest) to red (highest). The middle panel displays the interpolated MAGs using the posterior predictive means from the stacked posterior distribution in \eqref{eq: poterior density_stacked} derived from the continuous time model. We note that estimated MAGs along this path effectively capture the features of the observed MAGs. The right panel depicts the posterior predictive means of the latent process using \eqref{eq: poterior density_stacked} derived from the continuous time model. The variation in the right panel suggests that substantial spatial structure remains on the trajectory after accounting for ``Slope'' and ``NDVI''. The utility of the middle and right panels are distinct. The former is useful for understanding MAG as a feature associated with a path or trajectory. Our framework is able to predict the MAG at a completely arbitrary path, where no measurements have been measured, for a subject had an individual with a given set of personalized health attributes traversed that path. The latter, on the other hand, helps investigators glean lurking factors that may explain some of the residual spatial structure on a trajectory after accounting for the explanatory variables such as ``Slope'' or ``NDVI''.   

\begin{figure}[t]
    \centering
    \includegraphics[width=\textwidth]{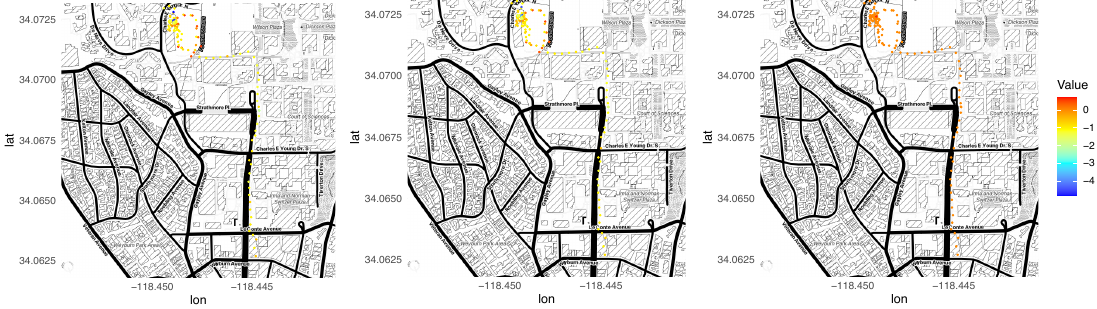}
    \caption{Left: 130 observed test data of $y$. Middle: 130 predictions of $y$ by the continuous time trajectory model. Right: 130 predictions of $z$ by the continuous time trajectory model.  \label{fig:pred_acti}}
\end{figure}

\begin{figure}[t]
    \centering
    \includegraphics[width=0.8\textwidth]{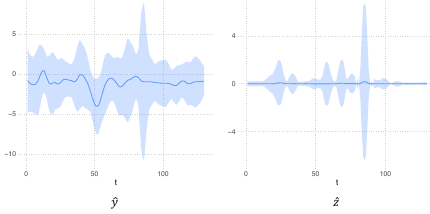}
    \caption{Predictive bands of $y$ (left) and $z$ (right) from continuous time trajectory model.\label{fig:band_acti}}
\end{figure}

Another key inferential element for spatial energetics is the estimate of an individual's daily profile of MAG, which allows medical professionals to recommend changes, if and as deemed appropriate, in the subject's daily mobility habits. As in the preceding figure, here, too, these daily profiles can be plotted for the outcome or for the residual. Figure~\ref{fig:band_acti} presents these plots. The left panel plots the posterior predictive mean and 95\% credible interval band for the MAG along the hours of the day to elicit the daily physical activity pattern of the subject and to better distinguish times of higher activity from those with lesser activity. The right panel presents the analogous plot for the residual process after accounting for trajectory effects represented by slope and greenness. 

Given the complications associated with streaming measurements at high resolutions from wearable devices, it is customary to encounter swaths of time intervals that have not recorded measurements either due to technical malfunction or user behavior. Our model based inferential framework uses the posterior predictive distributions to impute such missing values. The left panel in Figure~\ref{fig:infill_imputation} presents the reconstructed MAG for the subject under consideration at some epochs using the continuous time model. The right panel presents posterior predictive credible intervals that reveal the model's ability to effectively capture 130 held-out values for predictive validation. 

\begin{figure}[t]
    \centering
    \includegraphics[width=0.45\textwidth]{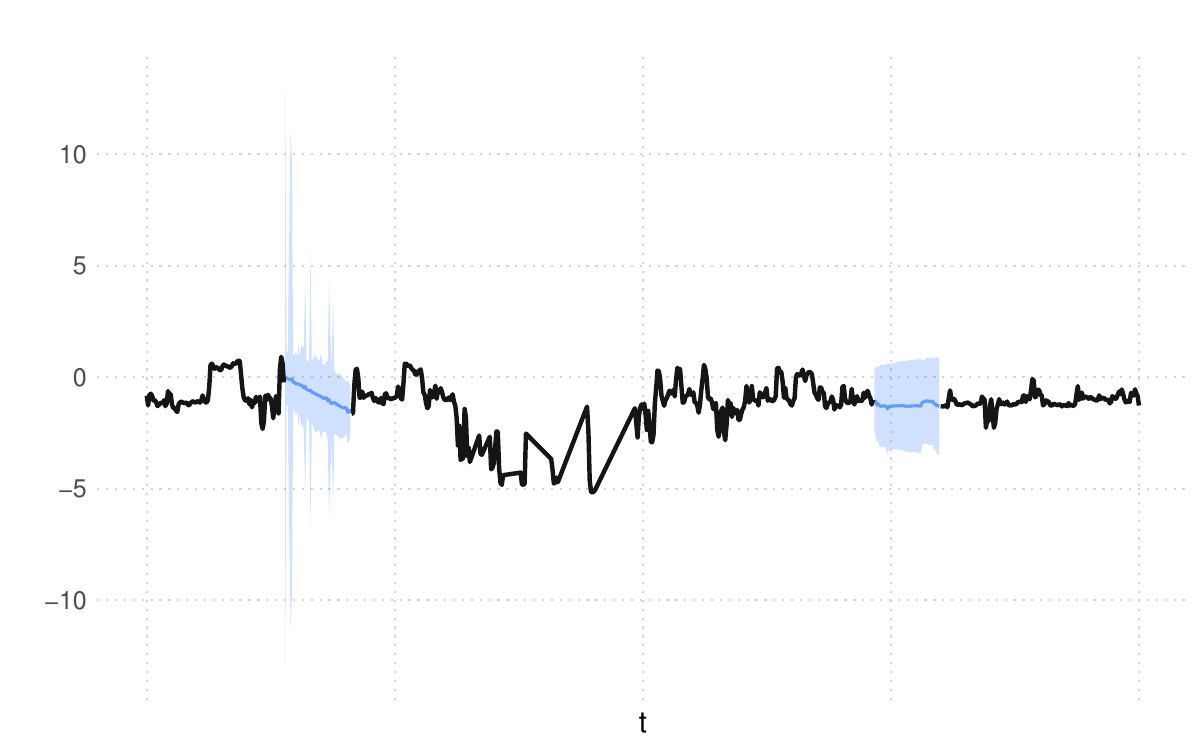}
    \includegraphics[width=0.45\textwidth]{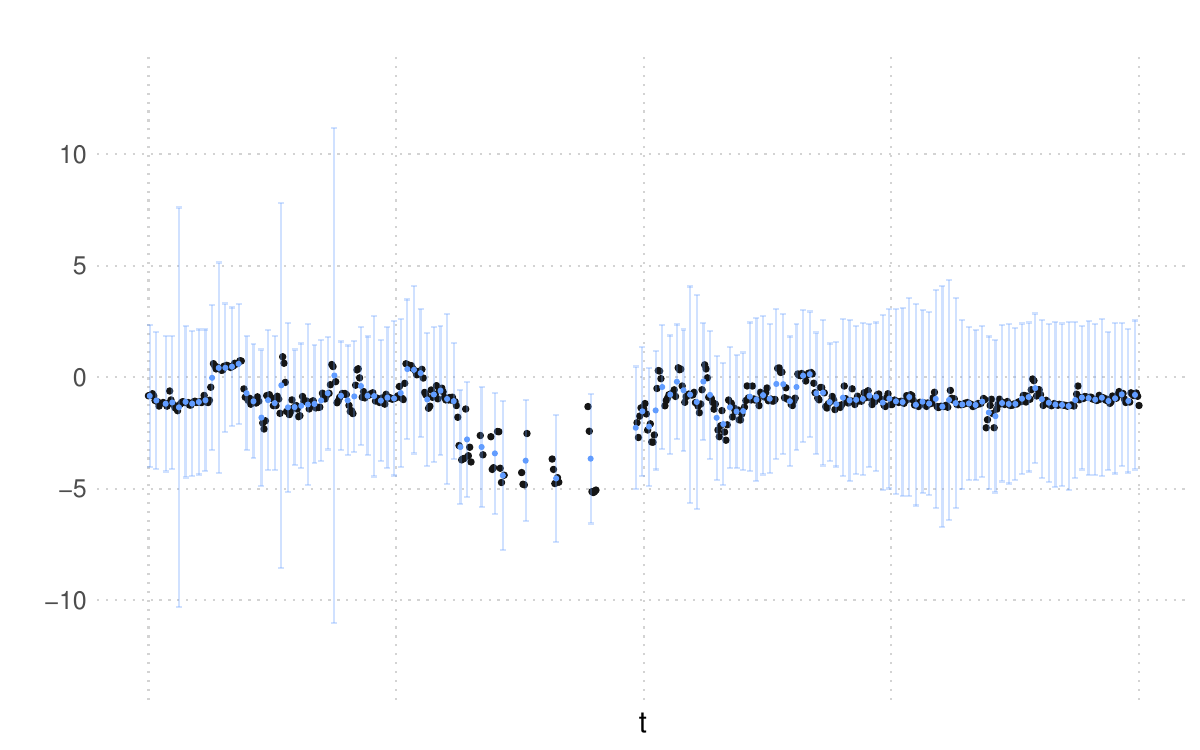}
    \caption{Left: Bayesian imputation of missing measurements at some epochs using the continuous time trajectory model. Right: 95\% credible intervals and posterior predictive means of 130 held-out MAG values using the continuous time trajectory model.  \label{fig:infill_imputation}}
\end{figure}

Finally, we compare the continuous time, discrete time and NSDLM using DIC and WAIC. We extract $150$ distinct time points from the above data, ensuring that they are equally spaced and compatible with the discrete-time trajectory model in \eqref{model:condis} and the NSDLM introduced in Section~\ref{sec:sim_proposed}. For stacking of \eqref{model:condis}, the candidates for the parameters were $\phi\in \{1, 1/10\}$, $\nu\in \{1, 1/3\}$ and $\delta_{\beta},\delta_z\in \{5,1/5\}$. Table~\ref{tab:app-fitting} reveals that the continuous time trajectory model is preferred (lower values) to the others in either of these metrics, while the discrete time model considerably outperforms NSDLM. In designing a physical activity recommending system that trains one model, our analysis suggests using the continuous time model is preferable although the discrete time model should also be competitive.

\begin{table}[t]
    \caption{Goodness of fit of the continuous time trajectory model, the discrete time trajectory model, and the NSDLM, measured by DIC and WAIC. \label{tab:app-fitting} }
    \centering
    \begin{tabular}{ccccc}
        \toprule
        \textbf{} & \textbf{Continuous} & \textbf{Discrete}  & \textbf{NSDLM} \\
        \midrule
        DIC  & -60.318 & -17.996 & 93.249  \\
        WAIC & -87.106 & -40.116 & 150.508   \\
        \bottomrule
    \end{tabular}
\end{table}

\section{Summary} \label{section:discuss}
We have devised a Bayesian inferential framework for spatial energetics that aims to analyze data collected from wearable devices containing spatial information over paths or trajectories traversed by an individual. Data analytic goals include estimating underlying spatial-temporal processes over trajectories that are posited to be generating the observations. A salient requirement for appropriately modeling spatial dependence in such applications is to model spatial locations as a function over time. We introduce such dependence in two broad classes of models: one that treats time as discrete and another that treats time as continuous. The former builds on Bayesian dynamic linear models and the latter employs spatial-temporal covariance functions to specify the underlying process. For conducting inference, we propose Bayesian predictive stacking as an effective method, where fully tractable conjugate posterior distributions up to certain parameters are assimilated, or stacked, to deliver Bayesian inference using a stacked posterior. Our framework offers some theoretical results to justify why predictive stacking renders effective posterior inference.                

\section*{Computer Programs}
Computer programs used in the manuscript for generating data for our simulation experiments in Section~\ref{section:simulation} and the application presented in Section~\ref{section:application} have been developed for execution in the \texttt{R} statistical computing environment. The programs are available for download in the publicly accessible GitHub repository \href{https://github.com/TomWaka/BayesianStackingSpatiotemporalModeling}{https://github.com/TomWaka/BayesianStackingSpatiotemporalModeling}.

\section*{Acknowledgments}
Tomoya Wakayama was supported by research grants 22J21090 from JSPS KAKENHI and JPMJAX23CS from JST ACT-X. Sudipto Banerjee was supported, in part, by research grants R01ES030210 and R01ES027027 from the National Institute of Environmental Health Sciences (NIEHS), R01GM148761 from the National Institute of General Medical Science (NIGMS) and DMS-2113778 from the Division of Mathematical Sciences (DMS) of the National Science Foundation.

\section*{Appendix}

\subsection*{Notation}
The notation $[A \mid B]$ represents a block matrix formed by horizontally concatenating $A\in \R^{p\times q}$ and $B\in \R^{p\times r}$. Given a $p\times q$ matrix $A$ and a $m\times n$ matrix $B$, $A \otimes B$ is the Kronecker product producing a $pm \times qn$ block matrix with $(i,j)$th block is $a_{ij} B$, while $A\oplus B$ denotes the block diagonal matrix with $A$ and $B$ along the diagonal. A $p$-dimensional random variable $x$ is said to be distributed as a multivariate t-distribution $t_{\nu}(\mu,\Sigma)$ with parameters $(\nu,\mu,\Sigma)$ if it has the density
\begin{equation*}
f(x; \mu, \Sigma, \nu) = \frac{\Gamma(\frac{\nu+p}{2})}{\Gamma(\frac{\nu}{2})(\nu\pi)^{p/2}\mathrm{det}(\Sigma)^{1/2}} \left(1+\frac{1}{\nu}(x-\mu)^{\top}\Sigma^{-1}(x-\mu)\right)^{-\frac{\nu+p}{2}},
\end{equation*}
where $\Gamma(\cdot)$ is the gamma function and $\mathrm{det}(\cdot)$ is the determinant of a matrix.

\subsection*{Proof of Lemma~\ref{lem:equivalence-t}}

\begin{proof}
Let $\mathbb{P}_t^*$ be the probability law endowed on a finite spatial realization of $y_t(s)$ with true parameters and let $\mathbb{P}_t'$ be that with parameters $(\sigma_*, \phi', \delta_z')$. For all $t=1,\ldots,T$, the Mat\'{e}rn based model without trend (i.e., $\beta=0$) is
\begin{align*}
     y_t(s) &=  z_t(s) + \eta_{1t}(s), \quad 
     z_t(s) \overset{\mathrm{ind}}{\sim} GP\left(0,\sigma^2 \Delta_{zt}^{2} K_{\phi}(\cdot,\cdot)  \right),
\end{align*}
where $\eta_{1t} = (\eta_{1t}(s_1),\ldots,\eta_{1t}(s_n))^{\T} \overset{\mathrm{i.i.d.}}{\sim} N(0,\sigma^2I_n)$ over any finite set of $n$ spatial locations within a bounded region, and $\Delta_{zt}^{2} = \sum_{j=1}^t \delta_z^{2j} (\alpha\sigma)^{2(j-1)}$. Applying Theorem 2.1 in \cite{tang2021identifiability}, we obtain that $\mathbb{P}_t'$ and $\mathbb{P}_{t}^*$ are equivalent. The equivalence of the joint distributions follows.
\end{proof}

\subsection*{Intuitive images of observational assumption}
All theorems in Section~\ref{sec: poscon} place the following restriction on the nature of $n$ spatial locations within a bounded region $\mS$:
\begin{equation*}
\max_{s \in \mS} \min_{1 \le i \le n} |s - s_i| \asymp n^{-\frac{1}{d}}.
\end{equation*}
Figure~\ref{fig:scatter} presents a depiction of the above assumption, which indicates that the spatial locations are scattered evenly.

\begin{figure}[t]
    \centering
        \includegraphics[width=0.7\textwidth]{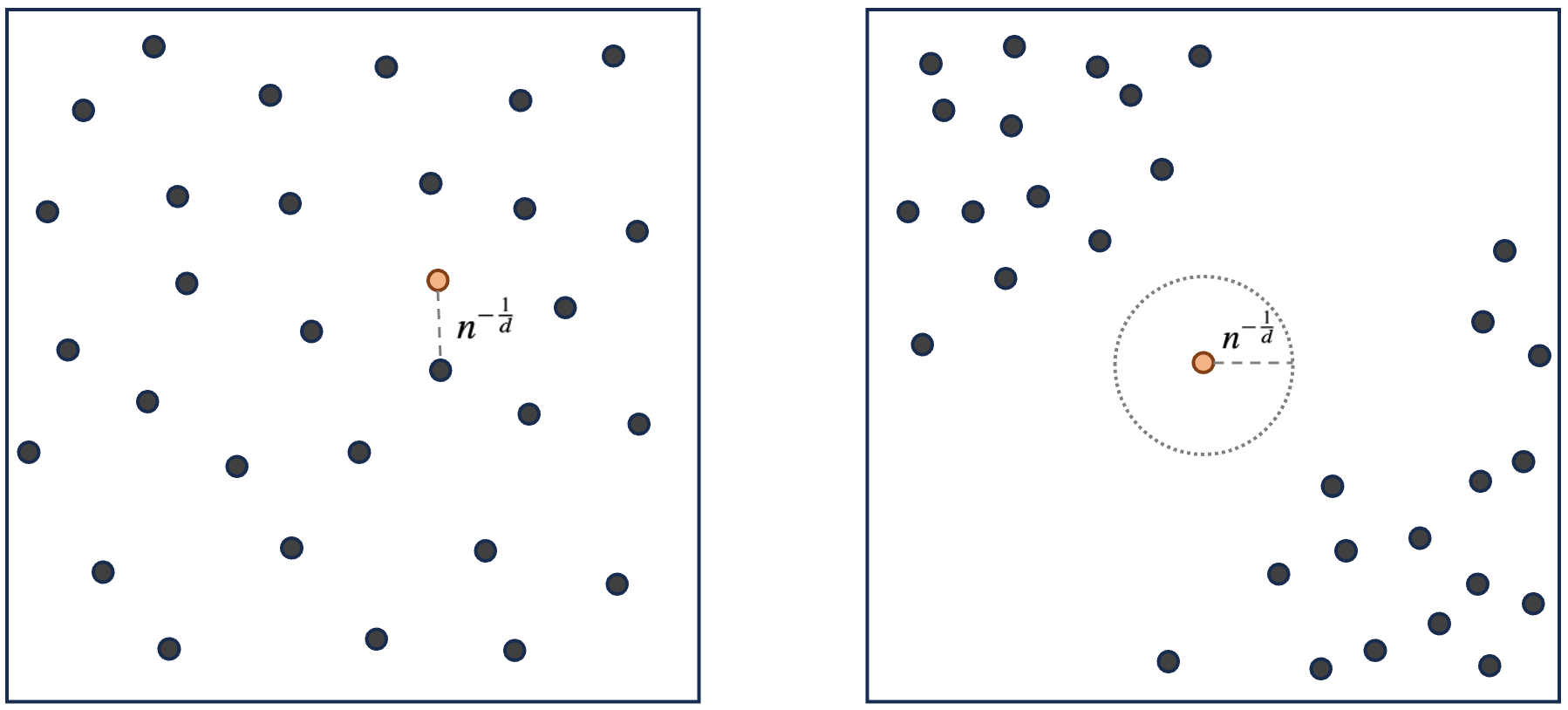}
    \caption{Left: scattered points such that $\max_{s \in \mS} \min_{1 \le i \le n} |s - s_i| \asymp n^{-\frac{1}{d}}$. Right: a polarized set of points on the domain such that $\max_{s \in \mS} \min_{1 \le i \le n} |s - s_i| > n^{-\frac{1}{d}}$. 
 \label{fig:scatter} }
\end{figure}

\subsection*{Proof of Theorem~\ref{thm:var}}

Before proceeding with the proof of Theorem~\ref{thm:var}, we recall the following lemma about the convergence of a random series.
\begin{lemma}\label{lem:asmcon}
Let $\{X_i:i\in \mathbb{N}\}$ be independent random variables with finite second moment and $\{a_i\in \mathbb{R}^{+} :i\in \mathbb{N}\}$ be an increasing positive number sequence such that $a_i \uparrow\infty$. If $\sum_{i=1}^{\infty} Var(X_i)/a_i^2<\infty$, it holds that
\begin{align*}
    \frac{\sum_{i=1}^n (X_i- \mathbb{E}[X_i] ) }{a_n} \to 0  \quad \mathrm{a.s.}.
\end{align*}
\end{lemma}
\begin{proof}
Let $Y_i = (X_i - \mathbb{E}[X_i] ) /a_i$. Since $\mathbb{E}[Y_i]=0$ and $\sum_{i=1}^{\infty} Var(Y_i)<\infty$, it follows from Kolmogorov's one series theorem (e.g., Theorem 2.5.6. in \cite{durrett2019probability}) that $\sum_{i=1}^{\infty} Y_i<\infty$ almost surely. Then, from Kronecker's lemma, $a_n^{-1}\sum_{i=1}^n (X_i- \mathbb{E}[X_i] )$ converges to $0$ almost surely.
\end{proof} 

We now present the main proof of Theorem~\ref{thm:var}.

\begin{proof}
We rewrite the notation introduced in Proposition~\ref{prop:post} as $n_t = n_{t-1}+n$, $n_ts_t = n_{t-1}s_{t-1} + ({y_t} - {f}_t )Q_t^{-1}({y_t} - {f}_t ), {f}_t =  \alpha {m}_{t-1}, Q_t =  R_t  + I_n, {m}_{t} = {f}_t + R_t  Q_t^{-1} ({y}_t-{f}_t), R_t = \alpha ^2 W_{t-1} + \delta_z^{'2} K_{\phi'} (\chi)$ and $W_t = R_t - R_t^{\T}Q_t^{-1}R_t$. Here, we prove $\frac{n_t s_t - n_{t-1}s_{t-1}}{n} \rightarrow t\sigma_*^2$ by mathematical induction on $t$. 

The base step $t=1$ yields 
\begin{align*}
    n_1s_1 - n_0s_0 &= ({y}_1 - \alpha {m}_0)^{\T} Q_1^{-1} ({y}_1 - \alpha {m}_0)\\
    & =  {y}_1 ^{\T} \left( (\alpha ^2 + \delta_z'^{2} )K_{\phi'} (\chi)+I_n\right)^{-1} {y}_1 
\end{align*}
Let $U_n$ be a unitary matrix such that $U_n K_{\phi'} U_n^{\T}$ is diagonal and let $\lambda_i^{(n)}$ be the $i$th eigenvalue of $K_{\phi'}$ for $i=1\ldots,n$. Since $U_n {y}_1$ follows a zero-centered multivariate normal distribution with a diagonal covariance matrix whose $i$th diagonal is $\sigma_*^2(1 + \delta_z'^2\lambda_i^{(n)})$ under $\mathbb{P}'$, we obtain 
\begin{equation*}
    n_1s_1 - n_0s_0 = \sum_{i=1}^n \frac{\sigma_*^2(1 + \delta_z'^2\lambda_i^{(n)}) }{1+ (\alpha ^2 + \delta_z'^{2}) \lambda_i^{(n)} } u_i^2,
\end{equation*}
where $u_i$ follows the standard normal distribution. Let $A_i =\left(\sigma_*^2(1 + \delta_z'^2\lambda_i^{(n)}) \right)\big/\left(1+ (\alpha ^2 + \delta_z'^{2}) \lambda_i^{(n)}\right) $. Because $\lambda_i^{(n)}\le C n i^{-2\nu/d - 1}$ for all $i=1,\ldots,n$ from Corollary 2 in \cite{tang2021identifiability}, $(\sum_{i=1}^n A_i )/ (n\sigma_*^2)$ converges to $1$ as $n\to\infty$. By $\sum_{i=1}^{\infty} A_i^2/i^2 <\infty $ and Lemma~\ref{lem:asmcon}, we obtain
\begin{equation*}
    \frac{n_1s_1 - n_0s_0}{n} = \frac1n \sum_{i=1}^n A_i u_i^2~\rightarrow~\sigma^2_*, \quad \mathbb{P}'-\text{a.s.}
\end{equation*}
Owing to the equivalence of $\mathbb{P}_*$ and $\mathbb{P}'$, $\frac{n_1s_1 - n_0s_0}{n} \rightarrow \sigma^2_*$ holds $\mathbb{P}_*$-almost surely. 

Then, as an inductive step, we assume that $\frac{n_ts_t}{n} \rightarrow t\sigma_*^2$ holds $\mathbb{P}_*$-almost surely and we consider the next period. We have
\begin{align*}
    n_{t+1}s_{t+1} - n_{t}s_{t}  = ({y}_{t+1} - \alpha {m}_t)^{\T} Q_{t+1}^{-1} ({y}_{t+1} - \alpha {m}_t),
\end{align*}
where $Q_{t+1} = (I_n + \alpha^2W_t + \delta_z'^2 K_{\phi'} )$ and $W_t = R_t - R_t^{\T}Q_t^{-1}R_t$.  Let $\lambda_{t,i}^{(n)}$ be the $i$th eigenvalue of $W_{t}$. Because $W_t = (R_t^{-1} + I_n)^{-1}$ and $R_t = \alpha^2 W_{t-1} +\delta_z'^2 K_{\phi'}$, it holds that, under $\mathbb{P}'$,
\begin{equation*}
    \lambda_{t,i}^{(n)} = \frac{ (\alpha^2+ \delta_z'^2)\lambda_{t,i}^{(n)} }{ 1 + (\alpha^2+ \delta_z'^2)\lambda_{t,i}^{(n)}}.
\end{equation*}
Note that $\lambda_{t,i}^{(n)} \to \left(\alpha^2+\delta_z'^2\right)^t\lambda_i^{(n)}$ as $i\to\infty$. Then,
\begin{align*}
    \frac1n({y}_{t+1} - \alpha {m}_t)^{\T} Q_{t+1}^{-1} ({y}_{t+1} - \alpha {m}_t)
    & = \frac1n {y}_{t+1}^{\T}U_n^{\T} U_n Q_{t+1}^{-1}U_n^{\T} U_n  {y}_{t+1}
      - \frac{2\alpha}{n} {y}_{t+1}^{\T} Q_{t+1}^{-1} {m}_t  
      + \frac{\alpha^2}{n} {m}_t^{\T} Q_{t+1}^{-1}  {m}_t
    \\
    & = \frac1n {y}_{t+1}^{\T}U_n^{\T} U_n Q_{t+1}^{-1}U_n^{\T} U_n  {y}_{t+1} +o(1) \\
    & = \frac1n \sum_{i=1}^n \frac{ \sigma_*^2 + v'^2\lambda_i^{(n)} }{ 1 + \alpha^2\lambda_{t,i}^{(n)} + \delta_z'^2\lambda_i^{(n)} } u_i^2  +o(1)
\end{align*}
where $u_i\sim N(0,1)$. The second equality holds because when we recursively expand $m_t$ by its definition, we find that $m_t$ includes a factor of $R_t$. Then, the second and third terms are $o(1)$ due to the eigenvalue decay of $R_t$. Hence, using Lemma~\ref{lem:asmcon} and the equivalence of the distributions, we obtain 
\begin{equation*}
    \frac1n \sum_{i=1}^n \frac{ \sigma_*^2 + v'^2\lambda_i^{(n)} }{ 1 + \alpha^2\lambda_{t+1,i}^{(n)} + \delta_z'^2\lambda_i^{(n)} } u_i^2 ~\rightarrow~ \sigma^2_*, \quad \mathbb{P}_*-\text{a.s.}
\end{equation*}
Therefore, we have
\begin{equation*}
    \frac{n_{t+1}s_{t+1}}{n} \rightarrow (t+1)\sigma_*^2, \quad \mathbb{P}_*-\text{a.s.}
\end{equation*}
This concludes the induction step and we conclude $\frac{n_ts_t - n_{t-1}s_{t-1}}{n} \rightarrow t\sigma_*^2$ holds for each $t=1,2,\ldots,T$. From Chebyshev's inequality, $p (\sigma^2 \mid \mD_{\chi_n,t}) \rightsquigarrow  \delta(\sigma_*^2)$ holds $\mathbb{P}_*$-almost surely.

\end{proof}\subsection*{Proof of Theorem~\ref{thm:ypred}}

\begin{proof}
   Recall that $n_t = n_{t-1}+n$, $n_ts_t = n_{t-1}s_{t-1} + ({y_t} - {f}_t )Q_t^{-1}({y_t} - {f}_t ), {f}_t =  \alpha {m}_{t-1}, Q_t =  R_t  + I_n, {m}_{t} = {f}_t + R_t  Q_t^{-1} ({y}_t-{f}_t), R_t = \alpha ^2 W_{t-1} + \delta_z^{'2} K_{\phi'} (\chi)$ and $W_t = R_t - R_t^{\T}Q_t^{-1}R_t$. We decompose the prediction error for the latent term $\tilde{z}_t(s_0)$ as
    \begin{align}
        \mathbb{E}_* [(Z_{tn}(s_0)-\tilde{z}_t(s_0))^2] &= \mathbb{E}_* [(Z_{tn}(s_0)- \mathbb{E}[\tilde{z}_t(s_0)\mid \mD_{\chi_n,t} ] + \mathbb{E}[\tilde{z}_t(s_0)\mid \mD_{\chi_n,t} ] -  \tilde{z}_t(s_0))^2] \nonumber \\
        &= \mathbb{E}_* [(Z_{tn}(s_0)- \mathbb{E}[\tilde{z}_t(s_0)\mid \mD_{\chi_n,t} ])^2 + (\tilde{z}_t(s_0)-\mathbb{E}[\tilde{z}_t(s_0)\mid \mD_{\chi_n,t} ] )^2] \nonumber \\
        & = \underbrace{\mathbb{E}_* [Var(Z_{tn}(s_0))]}_{E_{n,t}^A} + \underbrace{\mathbb{E}_* [(\tilde{z}_t(s_0)-\mathbb{E}[\tilde{z}_t(s_0)\mid \mD_{\chi_n,t} ] )^2]}_{E_{n,t}^B} \label{eq:decomp-z}
    \end{align}
Focusing on the first term, $E_{n,t}^A$ in \eqref{eq:decomp-z}, we note the following distributions,
\begin{align*}
    Z_{tn}(s_0) \mid \mD_{\chi_n,t},\sigma^2,{m}_t  &\sim N( {c}_0^{\T}C^{-1}{m}_t, \sigma^2\delta_z'^{2} - {c}_0^{\T}C^{-1}{c}_0 ),\\
    {m}_t\mid \sigma^2,{f}_t,R_t &\sim N({f}_t, \sigma^2 R_t(R_t + I_n)^{-2}R_t ),
\end{align*}
where $C = \sigma^2\delta_z'^2K_{\phi}(\chi_n)$ and ${c}_0=(\sigma^2\delta_z'^2K_{\phi}(s,s_0))_{s\in\chi_n}$. According to the law of total variance \citep{blitzstein2019introduction},
\begin{align*}
     Var(Z_{tn}(s_0))
    & =   \mathbb{E}[ Var(Z_{tn}(s_0)) \mid \mD_{\chi_n,t}, \sigma^2, m_t ] + Var( \mathbb{E}[Z_{tn}(s_0)\mid  \mD_{\chi_n,t}, \sigma^2, m_t] )
         \\
    & =  \sigma^2(\delta_z'^2 - \delta_z'^2 \{K_{\phi}(s,s_0)\}_{s\in\chi_n}^{\T} K_{\phi}^{-1}(\chi_n)\{K_{\phi}(s,s_0)\}_{s\in\chi_n} \\
    & \quad + \{K_{\phi}(s,s_0)\}_{s\in\chi_n}^{\T} K_{\phi}^{-1}(\chi_n)  R_t(R_t + I_n)^{-2}R_t K_{\phi}^{-1}(\chi_n)\{K_{\phi}(s,s_0)\}_{s\in\chi_n}^{\T} ).
\end{align*}
From Theorem~\ref{thm:var} we know that $p(\sigma^2 \mid \mD_{\chi_n,t}) \to \delta(\sigma^2_*)$ as $n\to\infty$ under $\mathbb{P}_*$. We also obtain
\begin{align*}
    E_{n,t}^A &= \sigma_*^2(\delta_z'^2 - \delta_z'^2 \{K_{\phi}(s,s_0)\}_{s\in\chi_n}^{\T} K_{\phi}^{-1}(\chi_n)\{K_{\phi}(s,s_0)\}_{s\in\chi_n} \\
    & \quad + \{K_{\phi}(s,s_0)\}_{s\in\chi_n}^{\T} K_{\phi}^{-1}(\chi_n)  R_t(R_t + I_n)^{-2}R_t K_{\phi}^{-1}(\chi_n)\{K_{\phi}(s,s_0)\}_{s\in\chi_n}^{\T} ).
\end{align*}
Furthermore, the second term in \eqref{eq:decomp-z} can be represented as
\begin{equation*}
    E_{n,t}^B = \mathbb{E}_* \left[\left(\tilde{z}_t(s_0)- \{K_{\phi}(s,s_0)\}_{s\in\chi_n}^{\T} K_{\phi}^{-1}(\chi_n) {m}_t \right)^2\right].
\end{equation*} 

The prediction error of $\tilde{y}_t(s_0)$ can be decomposed as 
\begin{align*}
        \mathbb{E}_* \left[(Y_{tn}(s_0)-\tilde{y}_t(s_0))^2\right] &= \mathbb{E}_* \left[(Y_{tn}(s_0)- \mathbb{E}[\tilde{y}_t(s_0)\mid \mD_{\chi_n,t} ] + \mathbb{E}[\tilde{y}_t(s_0)\mid \mD_{\chi_n,t} ] -  \tilde{y}_t(s_0))^2\right]  \\
        &= \mathbb{E}_* \left[(Y_{tn}(s_0)- \mathbb{E}[\tilde{y}_t(s_0)\mid \mD_{\chi_n,t} ])^2 + (\tilde{y}_t(s_0)-\mathbb{E}[\tilde{y}_t(s_0)\mid \mD_{\chi_n,t} ] )^2\right]  \\
        & = \mathbb{E}_* \left[Var(Y_{tn}(s_0))\right] + \mathbb{E}_* \left[\left(Y_t(s_0)-\mathbb{E}[\tilde{y}_t(s_0)\mid \mD_{\chi_n,t} ] \right)^2\right]
    \end{align*}
Then, we have
\begin{equation*}
    \mathbb{E}_* \left[(Y_{tn}(s_0)-\tilde{y}_t(s_0))^2\right] - E_{n,t}^A - E_{n,t}^B \overset{n \to \infty }{\longrightarrow} 2\sigma_*^2 .
\end{equation*}
\end{proof}
\subsection*{Proof of Theorem~\ref{thm:stack}}
\begin{proof}
    \begin{align*}
        \mathbb{E}_* \left[\left(\tilde{y}_t(s_0) - \sum_{g=1}^G a_g \mathbb{E}_g[ \tilde{y}_t(s_0)\mid \mD_{\chi_n,t} ] \right)^2\right] 
        &= \sigma_*^2 + \mathbb{E}_* \left[\left(\tilde{z}_t(s_0) - \sum_{g=1}^G a_g \mathbb{E}_g[ \tilde{z}_t(s_0)\mid \mD_{\chi_n,t} ] \right)^2\right] \\
        &=\sigma_*^2 + \mathbb{E}_* \left[\left\{\sum_{g=1}^G a_g\left(\tilde{z}_t(s_0) -  \mathbb{E}_g[ \tilde{z}_t(s_0)\mid \mD_{\chi_n,t} ]\right) \right\}^2\right].
    \end{align*}
By the Cauchy--Schwarz inequality, the second term is bounded by $G\left(\sum_g a_g^2\right)E_{n,t}^B$, which converges to zero.
\end{proof}

\subsection*{In-fill prediction of the discrete DLM model}
Here, we examine the in-fill predictive performance for the discrete time Bayesian DLM. First, we introduce the data-generating process along with the model in \eqref{model:dlm}. In the period $T = 20$, we uniformly sample $n = 50, 100, \ldots, 500$ spatial locations from the unit square $[0,1]^2\subset\R^2$ to generate training data (the left panel in Figure~\ref{fig:location-discrete}). The elements of initial values of the $2$-dimensional state vectors ${\beta}_0$ and ${z}_0$ are randomly generated from $N({0},4)$. With these initial values, we sequentially produce ${\beta}_t$ and ${z}_t$ using the autoregressive specification (see Section~\ref{sec: bayesian_dlm}) with $\delta_{\beta} = 1$, $\delta_z = 1$, $\sigma = 1$ and $K_{\phi}$ taken as the Mat\'ern kernel and $\phi=1/7$ and $\nu=1$. Here, matrices $G_{\beta,t}$ and $G_{z,t}$ are configured as identity matrices of sizes $p$ and $n$, respectively, and are considered time-invariant. Then, we sampled each element of $X_t$ from $N(0,4)$ and $y$ from \eqref{model:dlm}. In this setting, we consider the prediction of the one-step future data at all locations, including $n$ observed and $100$ newly sampled points, as shown in the right panel of Figure~\ref{fig:location-discrete}; the variation in the accuracy of spatial-temporal predictions with increasing spatial samples is of interest.

\begin{figure}[t]   
\centering
\includegraphics[width=0.8\textwidth]{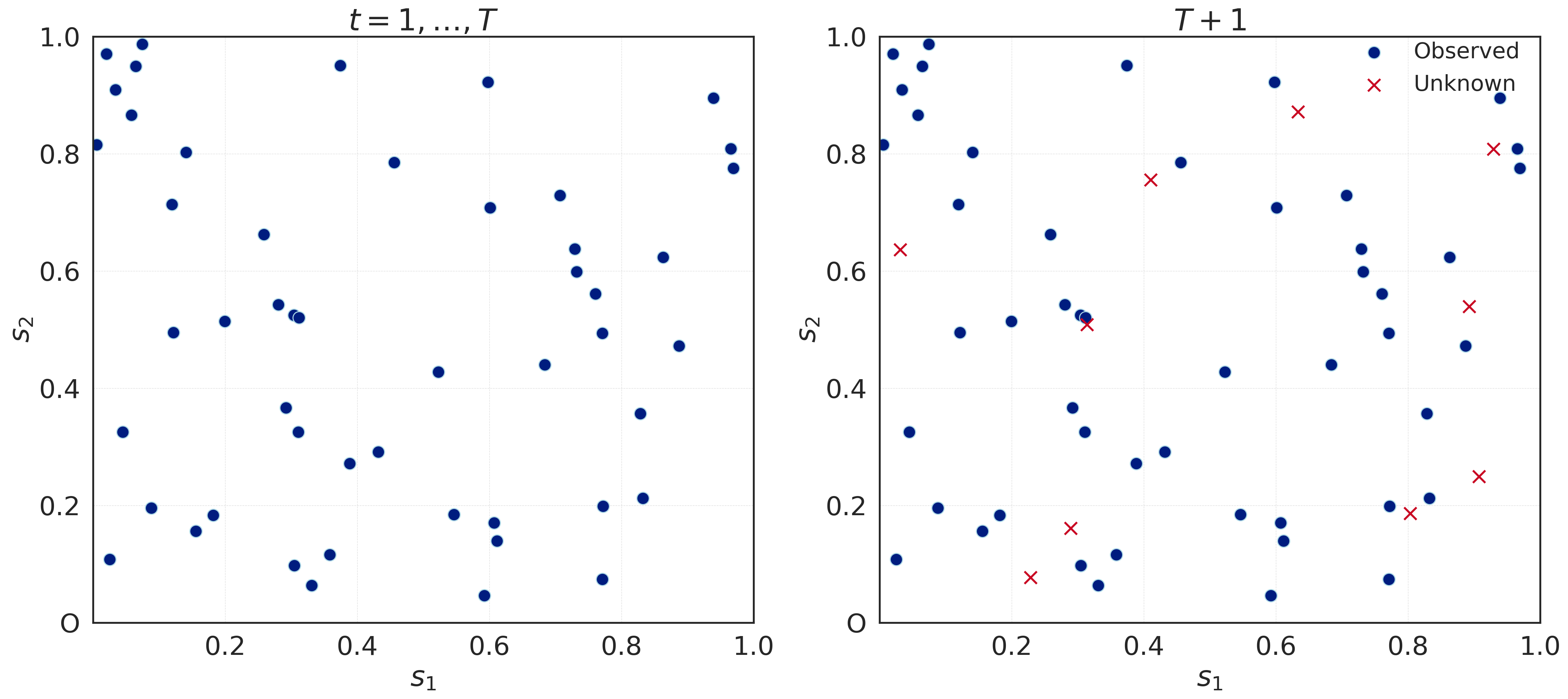}
\caption{Space-time prediction (discrete case). Left: observation locations at $t=1,..,T$. Right: prediction locations at $T+1$, including observation points and new points.\label{fig:location-discrete} }
\end{figure}
\begin{figure}[t] 
\centering 
\includegraphics[width=\textwidth]{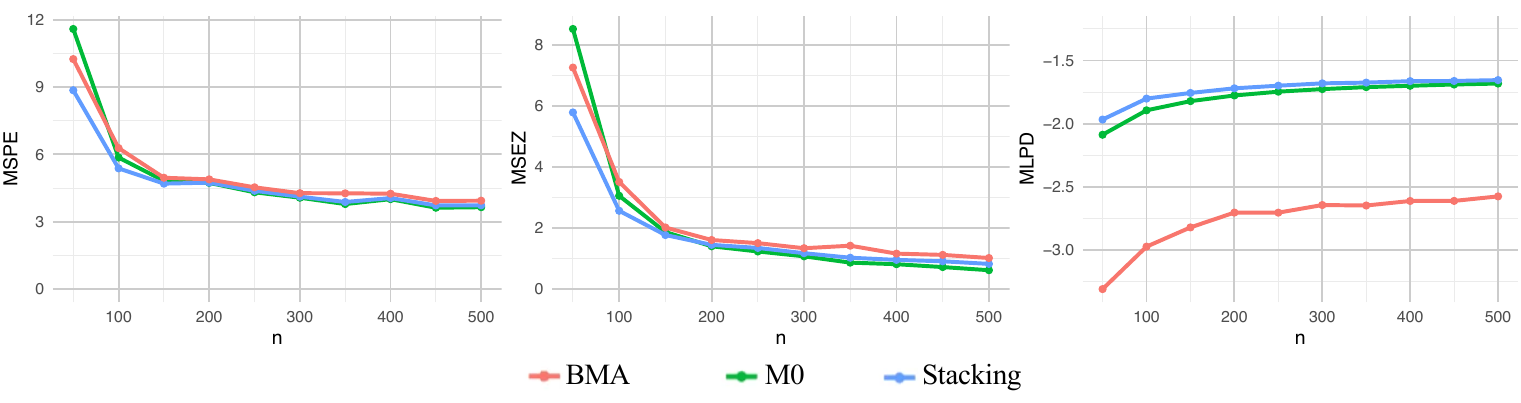}
\caption{Prediction performances of stacking (our proposed method), Bayesian model averaging (BMA), and $M_0$ (method with the oracle parameters) when $n$ grows from $50$ to $500$. Left: mean squared prediction error (MSPE). Middle: mean squared error for $z$ (MSE$z$). Right: mean log predictive density (MLPD). \label{fig:infill-review}}
\end{figure}

We generated $50$ different datasets using the above procedure and analyzed each dataset using the model in \eqref{model:dlm}. For predictive stacking, we defined a set of candidate parameters: $\phi\in \{1/5, 1/10\}$, $\nu\in \{2,1/2\}$, $\delta_{\beta}\in\{2,1/2\}$ and $\delta_z\in \{2,1/2\}$. We employed $20$-fold expanding window cross-validation and selected the stacking weights from $\Delta = \{ \{a_g\}_{g=1}^G \mid \sum_{g=1}^Ga_g=1, a_g\ge 0 \}$ to yield a simplex of the candidate predictions \eqref{eq: poterior density_stacked}. Furthermore, we implemented BMA with a uniform prior on candidate models and $M_0$, an oracle method with the true parameters assigned. As measures of performance, we adopted three metrics: mean squared prediction error (MSPE) and mean squared error for $z$ (MSE$z$) for stacking of means, and mean log predictive density (MLPD) for stacking of distributions. 

Figure~\ref{fig:infill-review} provides an overview of these predictions, where we report the average of the aforementioned metrics over the $50$ datasets. The left and center panels illustrate the enhancement in predictions of both the outcome and spatial effects in the in-fill paradigm with more observed locations. The right panel demonstrates that distributional stacking consistently improves with an increasing $n$, as indicated by the log predictive density. These findings indicate that a higher number of spatial points (as long as the points are dispersed) improves predictions. Furthermore, the stacking results are generally better than those of BMA, underscoring the significance of weight determination.

\subsection*{Discussion of \texorpdfstring{$E_{n,t}^A$}{ean} ~and~\texorpdfstring{$E_{n,t}^B$}{ebn}}
Here, we elaborate on the asymptotic behaviors of $E_{n,t}^A$ and $E_{n,t}^B$, introduced in Section~\ref{sec: poscon}.

First, because $E_{n,t}^A$ is analytically intractable, we numerically investigate its decay as $n$ increases from $50$ to $1600$. The observation locations are uniformly sampled from the unit square $[0,1]^2\subset \R^2$. We compute $E_{n,t}^A$ with $\delta_{\beta}=\delta_z=\sigma=1$, $(\phi, \nu)\in \{(1/2,1),(1/5,1/2), (1/10,1/3)\}$ and $T\in \{2,20\}$. The upper panels in Figure~\ref{fig:EAN} illustrate the predictive variance of $z$ in the absence of a trend, whereas the lower panel shows the predictive variance of $z$ when a trend is present. The training periods are $2$ and $20$ periods on the left and right sides, respectively. In all settings, the predicted variance decreases with an increasing sample size. The increase in $T$ indicates a faster decrease in predictive variance, suggesting that the rate is influenced by the training period $T$.

\begin{figure}[t]
    \centering
    \begin{minipage}{0.8\textwidth}
        \includegraphics[width=\textwidth]{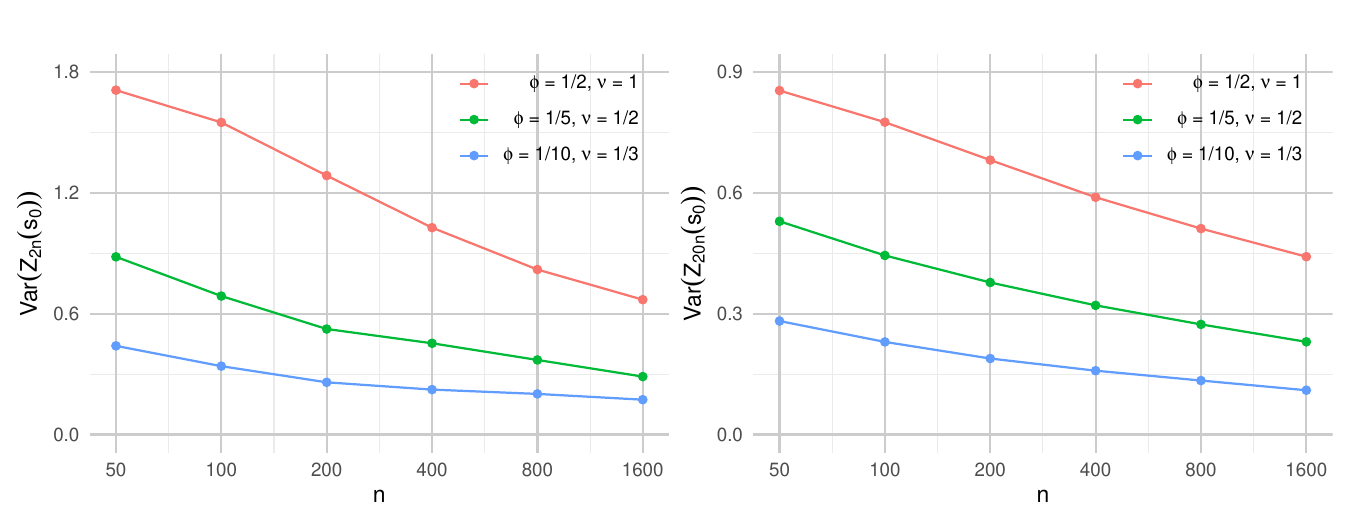}
    \end{minipage}
    \begin{minipage}{0.8\textwidth}
        \includegraphics[width=\textwidth]{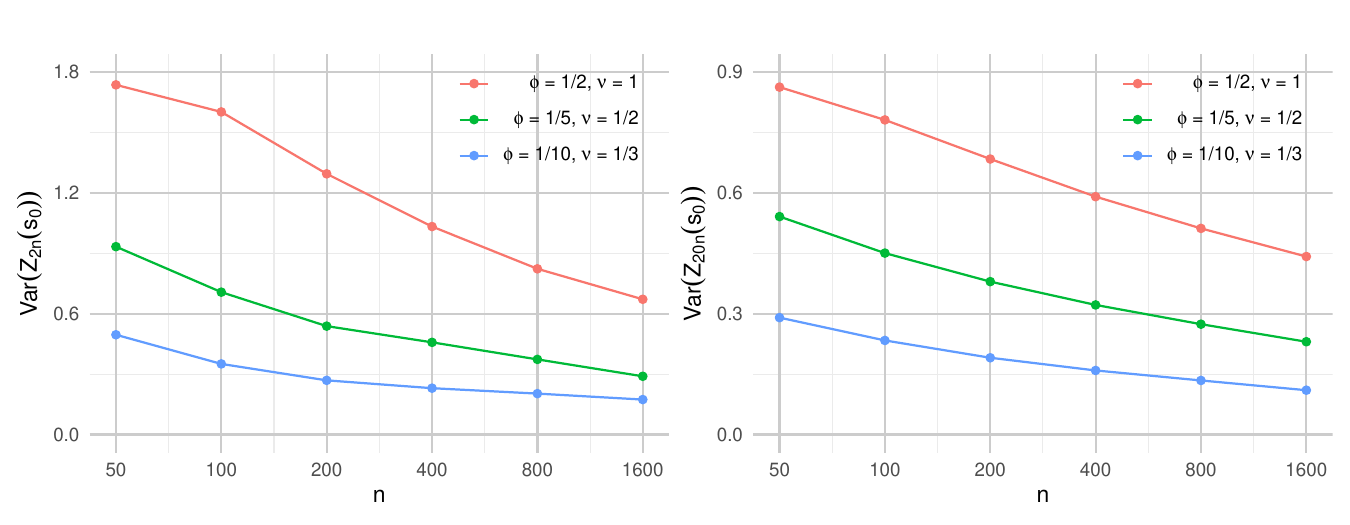}
    \end{minipage}
    \caption{Decay of $E_{n,t}^A$ with an increasing sample size. The upper panels show predictive variances of $z$ without a trend, and the lower panels describe those with a trend. The left column shows those when the training period is $2$, while the right column illustrates those when the training period is $20$. \label{fig:EAN} }
\end{figure}

Next, we consider the decay of $E_{n,t}^B$. For simplicity, we assume $T=1$ and the space has dimensions of $1$. Here, we denote $\chi_n = \{ i/n, i= -n, -n+1,\ldots, n-1, n \}\subset\mathbb{R}$, $\tilde{\chi}_n = \{ i/n, i\in \mathbb{N} \}\subset\mathbb{R}$. Let $\mD_{-0,t}$ be data in $\chi_n \setminus \{0\}$ until time $t$, $\tilde{\mD}_{-0,t}$ be data in $\tilde{\chi}_n \setminus \{0\}$ until time $t$, $e := \mathbb{E}_* [ \tilde{z}(0) - \mathbb{E}[ \tilde{z}(0) \given \mD_{-0,t} ] ],$ and $\tilde{e} := \mathbb{E}_* [ \tilde{z}(0) - \mathbb{E}[ \tilde{z}(0) \given \tilde{\mD}_{-0,t} ]]$. The attenuation of $E_{n,t}^B$ with larger $n$ is then justified based on the following result from \cite{zhang2023exact}

\begin{theorem}
Assume $|e - \tilde{e} | \to 0$ as $n\to \infty$. Then, the following holds as $n\to \infty$:
    \begin{equation*}
        E_{n,t}^B\to 0.
    \end{equation*}
\end{theorem}

\subsection*{Procedure of Bayesian Model Averaging}

Bayesian Model Averaging for predictors is given by
\begin{equation*}
p(\tilde{y}\given \mD) = \sum_{g=1}^G p_g(\tilde{y} \given \mD)p(M_g \given\mD),
\end{equation*}
where $\mD$ is the dataset, $M_g$ is the $g$-th candidate model, $G$ is the number of candidate models, and $p(M_g\given\mD)$ is the posterior probability of model $M_g$ given by
\begin{equation*}
p(M_g\given\mD) = \frac{p_g(\mD)p(M_g)}{\sum_{l=1}^G p(\mD)p(M_l)},
\end{equation*}
where $p(M_g)$ is the prior probability of model $M_g$ and $p_g(\mD)$ is the marginal likelihood of model $M_g$. If the prior is assumed to be uniform, $p(M_g)=1/G$ for $g=1,\dots,G$. 

\subsection*{Supplementary analysis of actigraph data}
For the analyses presented in Section~\ref{section:application}, we incorporated time varying estimates of ``Slope'' and ``NDVI''. We extracted $150$ points from the $650$ data points, as we did in Section~\ref{section:application}, to ensure the applicability of the discrete-time model. We then applied the continuous-time trajectory model, the discrete-time trajectory model, and the DLM to this subset. Figure~\ref{fig:app-beta} displays the posterior means and 95\% credible bands for the slopes obtained from each model. The continuous-time trajectory and discrete-time trajectory models show narrower credible intervals than the DLM. This indicates that the trajectory models provide more accurate slope estimates than the DLM since they account for spatial-temporal effects. We note that the aim of the current research requires accounting for these explanatory variables not only to improve the predictive inference presented in the manuscript, but also to infer about the underlying latent process posited to be generating the observations. Rather than identifying global statistical significance of explanatory variables on the subject's MAG, the time-varying impact of the explanatory variables enriches the predictive framework and better accounts for variation in the outcome, which, in turn, translates to improved estimation of the latent process as a spatially-temporally structured residual of the regression. 

\begin{figure}[t]
    \centering
    \includegraphics[width=\textwidth]{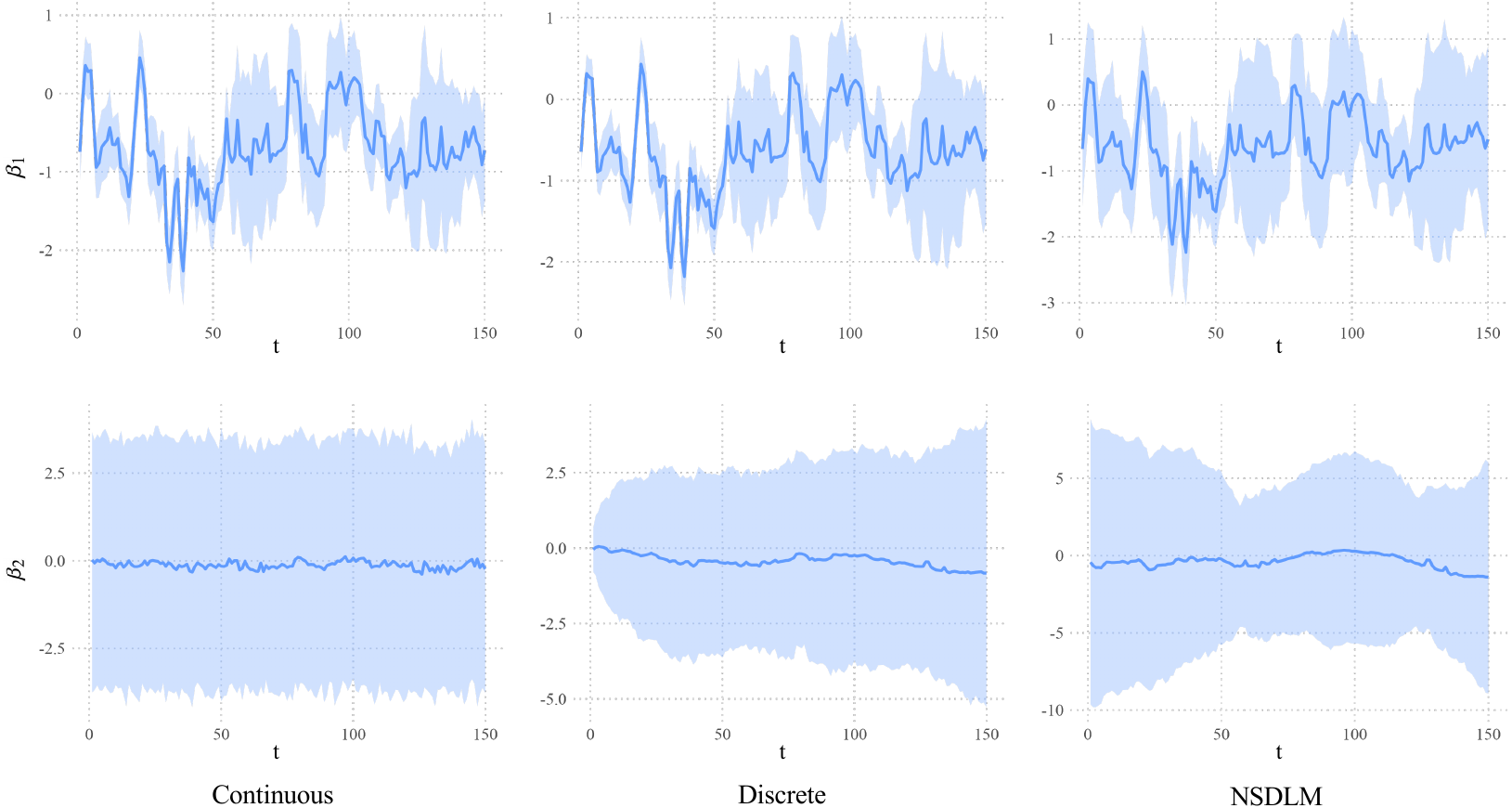}
    \caption{Two estimated slopes of each method. The solid lines and the shaded areas represent the posterior means and the 95\% credible intervals, respectively.\label{fig:app-beta}}
\end{figure}

\bibliographystyle{plainnat}
\bibliography{main}

\end{document}